\documentclass[preprint,3p,times]{elsarticle}
\usepackage{algorithm}
\usepackage[noend]{algpseudocode}
\usepackage[shortlabels]{enumitem}
\usepackage{amsmath,amssymb,amsfonts,amsthm,latexsym}
\usepackage{mathrsfs}
\usepackage{graphicx}
\usepackage{mathtools}
\usepackage[normalem]{ulem}
\usepackage{scalerel}
\usepackage[mathcal]{euscript}
\usepackage{xspace}
\usepackage{wasysym}

\usepackage{algorithm}
\usepackage[noend]{algpseudocode}
\algrenewcommand\algorithmicrequire{\textbf{Input:}}
\algrenewcommand\algorithmicensure{\textbf{Output:}}
\usepackage{amssymb,stackengine,scalerel}

\journal{Journal}

{\begin{description}[leftmargin = 0.2cm, labelsep = 0.2cm]}
  {\end{description}}

\makeatletter
\def\moverlay{\mathpalette\mov@rlay}
\def\mov@rlay#1#2{\leavevmode\vtop{%
		\baselineskip\z@skip \lineskiplimit-\maxdimen
		\ialign{\hfil$\m@th#1##$\hfil\cr#2\crcr}}}
\newcommand{\charfusion}[3][\mathord]{
	#1{\ifx#1\mathop\vphantom{#2}\fi
		\mathpalette\mov@rlay{#2\cr#3}
	}
	\ifx#1\mathop\expandafter\displaylimits\fi}
\makeatother
\newcommand{\cupdot}{\charfusion[\mathbin]{\cup}{\cdot}}

\newcommand{\mc}{\mathcal}

\newcommand{\lc}[1]{\ensuremath{_{#1}}} 
\renewcommand{\P}{\mathcal{P}}
\newcommand{\NF}[1]{\mathscr{F}(#1)} 
\newcommand{\ble}{\boldsymbol{<}}
\DeclareMathOperator{\Int}{Int}
\DeclareMathOperator{\Ext}{Ext}
\DeclareMathOperator{\MIN}{MIN}
\DeclareMathOperator{\MAX}{MAX}
\DeclareMathOperator{\PS}{\ensuremath{\partial}}
\DeclareMathOperator{\first}{first}
\DeclareMathOperator{\parent}{\mathit{parent}}

\DeclareMathOperator{\E}{E}

\newcommand{\slope}[2]{\ensuremath{\Delta(#1,#2)}}

\newcommand{\yVal}[2]{\ensuremath{y(#1,#2)}}
\newcommand{\blw}[2]{\ensuremath{\varpi\mathit{(#1,#2)=1}}} 
\newcommand{\pari}[2]{\ensuremath{\mathit{\varpi(#1,#2)}}} 

\newcommand{\SP}[1]{\ensuremath{(\mc{S},#1)}}

\newcommand{\SPx}[2]{\ensuremath{(\mc{S}\lc{|#2},#1)}}

\newcommand{\area}[1]{\ensuremath{\mathrm{area}(#1)}}
\newcommand{\BST}{\ensuremath{\mathsf{BST}}\xspace}
\newcommand{\Hasse}{\mathfrak{H}}

\newcommand{\Ih}[1]{\ensuremath{I[#1)}}
\newcommand{\Ic}[1]{\ensuremath{I[#1]}}
\newcommand{\Io}[1]{\ensuremath{I(#1)}}

\usepackage[colorlinks]{hyperref}
\usepackage{cleveref}

\newtheorem{theorem}{Theorem}[section]
\newtheorem{lemma}[theorem]{Lemma}
\newtheorem{proposition}[theorem]{Proposition}
\newtheorem{corollary}[theorem]{Corollary}
\theoremstyle{definition}
\newtheorem{definition}[theorem]{Definition}
\newtheorem{obs}[theorem]{Observation}

\newtheorem{remark}[theorem]{Remark}

\crefname{theorem}{Theorem}{Thm.}
\Crefname{theorem}{Theorem}{Theorems}
\crefname{lemma}{Lemma}{Lemmas}
\Crefname{lemma}{Lemma}{Lemmas}
\crefname{proposition}{Prop.}{Prop.}
\Crefname{proposition}{Proposition}{Propositions}
\crefname{corollary}{Cor.}{Cor.}
\Crefname{corollary}{Corollary}{Corollaries}
\crefname{remark}{Remark}{Remark}
\Crefname{remark}{Remark}{Remark}
\crefname{definition}{Def.}{Def.}
\Crefname{definition}{Definition}{Definitions}
\crefname{figure}{Fig.}{Fig.}
\Crefname{figure}{Figure}{Figures}
\crefname{equation}{Equ.}{Equ.}
\Crefname{equation}{Equation}{Equations}
\crefname{example}{Example}{Exa.}
\Crefname{example}{Example}{Examples}
\crefname{algorithm}{Algorithm}{Alg.}
\Crefname{algorithm}{Algorithm}{Algorithms}
\crefname{obs}{Obs.}{Obs.}
\Crefname{obs}{Observation}{Observations}
\crefname{line}{Line}{Line}
\Crefname{line}{Line}{Lines}
\crefname{section}{Section}{Sec.}
\Crefname{section}{Section}{Sections}
\crefname{enumi}{Stat.}{Stat.}
\Crefname{enumi}{Statement}{Statements}
\renewcommand{\theenumi}{(\arabic{enumi})}
\renewcommand{\theenumii}{(\alph{enumii})}

\makeatletter
\renewcommand\p@enumii{\theenumi}
\renewcommand\p@enumiii{\theenumi\theenumii}
\makeatother

\begin{document}

\begin{frontmatter}

\title{Nesting of Touching Polygons}



\cortext[cor1]{corresponding author}
\author[LEI2,LEI1]{Carsten R.\ Seemann}
\ead{carsten@bioinf.uni-leipzig.de}
\author[LEI,MIS,VIE,UNAC,SFI]{Peter F.\ Stadler}
\ead{studla@bioinf.uni-leipzig.de}
\author[STO]{Marc Hellmuth\corref{cor1}}
\ead{marc.hellmuth@math.su.se}

\address[LEI2]{Swarm Intelligence and Complex Systems Group, Department of Computer Science, 
	Leipzig University, Augustusplatz	10, D-04109 Leipzig, Germany}

\address[LEI1]{Bioinformatics Group, Department of Computer Science and
  Interdisciplinary Center for Bioinformatics, Leipzig University,
  H{\"a}rtelstra{\ss}e 16-18, D-04107 Leipzig, Germany}
	

\address[LEI]{Bioinformatics Group, Department of Computer Science;
  Interdisciplinary Center for Bioinformatics; German Centre for
  Integrative Biodiversity Research (iDiv) Halle-Jena-Leipzig; Competence
  Center for Scalable Data Services and Solutions Dresden-Leipzig; Leipzig
  Research Center for Civilization Diseases; and Centre for Biotechnology
  and Biomedicine, Leipzig University, H{\"a}rtelstra{\ss}e 16-18,
  D-04107 Leipzig, Germany}
  
\address[MIS]{Max Planck Institute for Mathematics in the Sciences,
  Inselstra{\ss}e 22, D-04103 Leipzig, Germany}

\address[VIE]{Institute for Theoretical Chemistry, University of Vienna,
  W{\"a}hringerstra{\ss}e 17, A-1090 Wien, Austria}

\address[UNAC]{Facultad de Ciencias, Universidad Nacional de Colombia,
  Sede Bogot{\'a}, Colombia}

\address[SFI]{The Santa Fe Institute, 1399 Hyde Park Rd., Santa Fe, NM
  87501, United States}

\address[STO]{Department of Mathematics, Faculty of Science, 
	Stockholm University, SE - 106 91 Stockholm, Sweden}

\begin{abstract}
  Polygons are cycles embedded into the plane; their vertices are
  associated with $x$- and $y$-coordinates and the edges are straight
  lines. Here, we consider a set of polygons with pairwise non-overlapping
  interior that may touch along their boundaries.  Ideas of the \emph{sweep
  line algorithm} by Bajaj and Dey for non-touching polygons are adapted to
  accommodate polygons that share boundary points.  The algorithms
  established here achieves a running time of $\mathcal{O}(n+N\log
  N)$, where $n$ is the total number of vertices and $N<n$ is the total
  number of ``maximal outstretched segments'' of all polygons. It is
  asymptotically optimal if the number of maximal outstretched segments per
  polygon is bounded.  In particular, this is the case for convex polygons.
\end{abstract}

\begin{keyword}
  Polygon \sep Planar Graph \sep Characterization \sep Recognition Algorithm
  
  \MSC{05C10}
\end{keyword}

\end{frontmatter}

\sloppy

\section{Introduction}

Polygons are cycle graphs that are embedded in the plane as Jordan curves
such that all edges are straight lines. Hence, the polygon is determined
completely by the $x$- and $y$-coordinates of the vertices (i.e., the
corners of the polygon), together with their adjacency.  We are interested
here in the problem of determining the nesting on a given set of polygons
with non-overlapping interior that are, however, allowed to touch each
other on their boundaries. An example are contour lines, which touch each
other in saddle points or along edges that represent vertical cliffs. Since
the polygon set is overlap-free by assumption, their relative locations are
determined completely by a nesting tree, in which the descendants of
polygon $P$ are exactly the polygons located inside of $P$. The problem of
determining the nesting tree of polygons arises naturally e.g.\ in layered
manufacturing \cite{Choi2004}, in the analysis of contour data sets
\cite{Johnstone2017} or in graph theory \cite{SMSH:21}.  Polygons
  nested within a polygon can also be viewed ``holes'', an interpretation
  that has been studied as Disassociative Area Model (DAM)
  \cite{Kirby:89}.

The special case of overlap-free polygons that do even not
touch on their boundaries was considered already by
\citet{Bajaj:90}, with further improvements described by
\citet{Zhu1994}. These authors established a ``sweep line algorithm'' that
computes the nesting in $\mathcal{O}(n+N\log N)$, where $n$ is the total
number of vertices and $N<n$ is the total number of so-called ``maximal
outstretched segments'' of all polygons.  A related problem is the 
nesting of so-called ``connected components'' generated by a set of
polygons. Here, only polygons that are outer boundaries of connected
subgraphs are taken into account. For this problem, an $\mathcal{O}(n\log
n)$-time algorithm is described by \cite[Chapt.\ 4]{Zhu1994}, where $n$ is
the total number of vertices of all polygons which have distinct $x$- or
$y$-coordinates.

In this contribution, we investigate in \cref{sec:osp} ``outstretched''
segments of polygons, such that we can use ``maximal'' outstretched
segments in \cref{sec:sl} to determine whether a point lies inside or
outside of a given polygon.  Then, in \cref{sec:invest}, we investigate a
set of overlap-free polygons and their maximal outstretched segments. The
observations motivates an ordering of the maximal outstretched segments in
\cref{sec:order}, and characterizes how two polygons are nested.
The main results of this contribution is in \cref{sec:nest-forest}.  Here,
we show that the nesting problem for general overlap-free polygons can be
solved in $\mathcal{O}(n+N\log N)$ operations, where $n$ is the total
number of vertices and $N$ is the total number of maximal outstretched
segments of all polygons. This removes the restriction to disjoint polygons
in the algorithms of \citet{Bajaj:90,Zhu1994}. In particular, the
running-time for a set of overlap-free polygons, where each polygon is
convex, has bounded length or bounded number of maximal outstretched
segments, is optimal.

Throughout this contribution we assume that the set of polygons is given
explicitly, i.e., for each each polygon we are given a list of vertices and
edges of a graph together with the embedding coordinates of vertex in the
plane. In particular, the vertices and edges that are shared by
multiple polygons also appear multiple times.

\section{Preliminaries}\label{sec:prel}

\paragraph{\textbf{Basics}}
For a set $e=\{x,y\}$ of two points $x,y\in\mathbb{R}^2$, we denote by
$s[e]$ the straight line segment between $x$ and $y$, i.e.,
$s[e]\coloneqq\big\{(1-\lambda)\cdot x+\lambda\cdot y\mid \lambda\in
[0,1]\subseteq \mathbb{R}\big\}$.
By definition, $x,y\in s[x,y]$. Moreover, the interior of straight line
segment is denoted by $s(e)\coloneqq s[e]\setminus\{x,y\}$.  For
simplicity, we write $s(x,y)$ and $s[x,y]$ instead of $s[\{x,y\}]$ and
$s(\{x,y\})$, respectively.  Two sets $A$ and $B$ \emph{overlap},
whenever $A\cap B\not\in \{A,B,\emptyset\}$.  A collection $\mc{S}$ of sets
is \emph{overlap-free} if no two elements in $\mc{S}$ overlap.

\paragraph{\textbf{Graphs}}
We consider (undirected) graphs $G=(V,E)$ with finite vertex set $V(G)=V$
and edge set $E(G)=E\subseteq {V\choose 2}$, i.e., without loops and
multiple edges.  A tree $T$ is an acyclic connected graph, and it is rooted
if there is distinguished vertex $\rho_T$, called the root of $T$. A rooted
forest is a graph whose connected components are rooted trees. Note a tree
is a forest consisting of a single connected component. Given a rooted
forest $F$, we can define a partial order $\preceq_F$ on $V(F)$ by putting
$v\preceq_F w$ whenever there is a connected component $T$ of $F$ that
contains $v$ and $w$ and where $v$ lies one the unique path connecting
$\rho_T$ and $w$.  In this case, we say that $v$ is an ancestor of $w$.
Note that $\rho_T$ is always $\preceq_F$-minimal for all connected
components $T$ of $F$. A vertex $v$ is a \emph{child} of $w$ (or,
equivalently, $w$ is a \emph{parent} of $v$) in $F$ if $\{v,w\}\in F$ and
$w\preceq_F v$. Two vertices are \emph{siblings} if they have a common
parent or if they are the roots of two distinct connected components of
$F$.

Given a collection of sets $\mc S$, the \emph{Hasse diagram}
$\Hasse\coloneqq \Hasse{}[\mc{S}]$ \emph{of} $\mc{S}$ \emph{(w.r.t.\ set
inclusion)} is a graph whose vertices are the elements in
$\mc{S}$. Inclusion-maximal elements in $\Hasse$ are called \emph{roots},
and there is an edge $\{A_1,A_2\}\in E(\Hasse)$ if and only if,
$A_i\subsetneq A_j$ with $\{i,j\}=\{1,2\}$ and there is no $C\in\mc{S}$
with $A_i\subsetneq C\subsetneq A_j$.  In this context, it is well-known
that the Hasse diagram of an overlap-free set $\mathcal{S}$ must be a
rooted forest. For instance, this is a direct consequence of
\cite[Thm.~3.5.2]{Semple:03}.

\paragraph{\textbf{Planarity and Polygons}}
In the following, $\varphi\colon V\to \mathbb{R}^2$ denotes an injective
map of the vertices of $G=(V,E)$ into the Euclidean plane. Given such a map
$\varphi$, a \emph{(straight-line) embedding $\PS(G)$} of a graph $G=(V,E)$
is defined as $\PS(G)\coloneqq \big\{\varphi(v)\mid v\in V\big\}\cup
\big\{s[\varphi(u),\varphi(v)] \mid \{u,v\}\in E\big\}$. In simple words,
an embedding $\PS(G)$ associates each vertex with a unique point in the
plane and every edge by a straight-line connecting its endpoints. Given a
subgraph $H$ of $G$ with embedding $\PS(G)$, we denote with
$V(\PS(H))\coloneqq \{\varphi(v)\mid v\in V(H)\}$ the set of vertices and
with $E(\PS(H))\coloneqq \left\{\{\varphi(v),\varphi(u)\} \mid \{u,v\}\in
E(H) \right\}$ the set of edges of $\PS(H)$. The latter, in particular,
allows us to treat $\PS(H)$ as a graph theoretical object to which all the
standard terminology of graphs applies.
In the following, $v_x$ and $v_y$ (or, $u_x$ and $u_y$) denote the $x$- and
$y$-coordinate of the image $\varphi(v) = u \in \mathbb{R}^2$ of a vertex
$v$. 
For $W\subseteq V(G)$, we defined the following real-valued intervals:
\begin{align*}
  \Ih{W}\coloneqq [\min_{v\in W} v_x,\max_{x\in W} v_x),\hspace{.5cm}
  \Io{W}\coloneqq (\min_{v\in W} v_x,\max_{x\in W} v_x),\, and \hspace{.5cm}
  \Ic{W}\coloneqq [\min_{v\in W} v_x,\max_{x\in W} v_x].
\end{align*} 
In particular, since for an edge $e\in E(G)$ we have $e\subseteq V(G)$, the
intervals $\Ih{e}$, $\Io{e}$ and $\Ic{e}$ are well-defined.  In case $W=
V(G)$, we write $\Ih{G}$, $\Io{G}$ and $\Ic{G}$ instead of $\Ih{W}$,
$\Io{W}$ and $\Ic{W}$, respectively.

An edge $e =\{u,v\} \in E$ of $G$ is \emph{vertical} if $u_x=v_x$.  If $e$
is non-vertical, then $\min(e)$ denotes the (unique) vertex $w\in
e=\{u,v\}$ with $w_x = \min\{u_x,v_x\}$, and $\max(e)$ is the (unique)
vertex $w\in e=\{u,v\}$ with $w_x = \max\{u_x,v_x\}$. Note that
$\min(e)_x<\max(e)$ for non-vertical edges and thus, $\Ih{e}$, $\Io{e}$,
$\Ic{e}$ are non-empty in this case.  For vertical edges, we leave
$\min(e)$ and $\max(e)$ undefined, and we have $\Ih{e}=\emptyset$. By the
latter arguments, $\Ih{e}=\emptyset$ if and only if $e$ is a vertical edge.

An embedding $\PS(G)$ of $G$ is \emph{planar} if $s(\varphi(u),\varphi(v))$
for every edge $\{u,v\}\in E$ does not contain any point of the straight
line $s[\varphi(u'),\varphi(v')]$ for every edge $\{u',v'\}\in E\setminus
\{u,v\}$ \cite{fary1948straight}. A graph is \emph{planar} if it admits a
planar embedding. Due to F{\'a}ry's Theorem \cite{fary1948straight}
definition of planarity is equivalent to the ``usual'' definition of planar
graphs, where the planar embedding is defined in terms of so-called Jordan
curves, see e.g.\ \cite{Diestel2010}.

\begin{remark}
  From here on, all embeddings of all graphs are considered to be planar.
\end{remark}

A \emph{polygon} $P$ refers to the embedding of an cycle $C_P$ of $G$,
i.e., $P=\PS(C_P)$, where $C_P$ is a connected subgraph of $G$ such that
all vertices have degree two in $C_P$. This type of polygons is usually
called ``simple'' polygon. A \emph{corner} of a polygon $P$ refers to a
point $p\in V(P)$.  Note that corners may have the same $x$- or
$y$-coordinate as other corners, e.g., we may have corners with coordinates
$(v_x,v_y)$, $(v'_x,v_y)$ and $(v''_x,v_y)$. Thus, corners may lie on the
same straight line.  We refer to the bounded region enclosed by a polygon
$P$ as its \emph{interior} $\Int(P)$ and assume that $P \cap \Int(P) =
\emptyset$. The \emph{exterior} of $P$ is the unbounded region
$\Ext(P)\coloneqq \mathbb{R}^2\setminus(P\cup\Int(P))$.

\paragraph{\textbf{Sets of Polygons and Nesting Forest}}
In the following,  $\P$ denotes a set of polygons and put
$\Int(\P) \coloneqq \{\Int(P)\mid P\in \P\}$. 
We say that $\P$ is \emph{overlap-free} if
$\Int(\P)$ is overlap-free.  The Hasse diagram
$\Hasse[\Int(\P)]$ of $\Int(\P)$ with respect to inclusion thus is a rooted
forest $\NF{\P}$, which we call the \emph{nesting forest} of $\P$. For
simplicity, we identify vertices $\Int(P)$ in $\mc{F}(\P)$ by the
corresponding polygons $P\in \P$ and thus, by slight abuse of notation,
assume that $V(\NF{\P})=\P$.  Moreover, we say that a polygon $P$ is inside
of a polygon $P'$ if $\Int(P)\subseteq\Int(P')$, and the two polygons
\emph{touch} if $\Int(P)\cap\Int(P')=\emptyset$ but $P\cap P'\neq
\emptyset$.

\section{Outstretched Segments and Polygons}
\label{sec:osp} 

Given a subgraph $H$ of $G$ with embedding $\PS(G)$, a \emph{segment $S$
(of $\PS(H)$)} refers to the embedding of some subpath $Q_S$ of $H$, i.e.,
$S=\PS(Q_S)$. We say that a vertex $v\in V(S)$ is \emph{terminal (in $S$)}
if $v$ has degree one in $S$. Analogously, an edge $e =\{u,v\}\in E(S)$ is
\emph{terminal (in $S$)} if one of its vertices $u$ or $v$ is terminal in
$S$. For a segment $S$, we define $\MIN(S)$ (resp., $\MAX(S)$) as the set
of vertices $v$ for which $v_x$ is minimal (resp., maximal) among all
vertices in $S$. If $|\MIN(S)|=1$ or $|\MAX(S)|=1$, then we write $\min(S)$
or $\max(S)$ for the unique minimal or maximal element. Note that
$\min_{v\in S} v_x = v'_x$ for all $v'\in\MIN(S)$ and $\max_{v\in S} v_x =
v'_x$ for all $v'\in\MAX(S)$. 

\begin{definition} \label{def:prop-O}
  A segment $S$ satisfies Property~(O) if for all distinct 
  edges $e,f\in E(S)$ it holds that $\Ih{e}\cap \Ih{f}=\emptyset$.
\end{definition}
Note, for all edges $e\in E(S)$, we have $\Ih{e}\subseteq \Ih{S}$.
Moreover, $\Ih{S}=\emptyset$ implies that $\Ih{e}=\emptyset$ for all $e\in
E(S)$.  Since $\Ih{e} = \emptyset$ precisely if $e$ is vertical, the
intervals $\Ih{e}$ of non-vertical edges $e$ of $S$ partition the interval
$\Ih{S}$ whenever $S$ satisfies (O).

\begin{lemma}\label{lem:O}
  For a segment $S$, the following three conditions are equivalent:
  \begin{description}[nolistsep,leftmargin=!,labelwidth=.7cm]
  \item[(O)] $S$ satisfies Property~(O).
  \item[(O')]  
  \begin{enumerate}[nolistsep,noitemsep]
    \item[(a)] $\min_{v\in S}v_x$  and
      $\max_{v\in S}v_x$ correspond to the
      x-coordinates of the terminal vertices of $S$, and
    \item[(b)] $u_x\leq v_x$ for all vertices $u,v \in V(S)$ 
      with $u\preceq_S v$, where  $\preceq_S$ is 
      defined on $V(S)$ by choosing a vertex $v\in\MIN(S)$ 
      as the root of $S$.
  \end{enumerate}
  \item[(O'')]For all $\xi\in \Ih{S}$, there is exactly
    one edge $e \in E(S)$ with $\xi\in \Ih{e}$.
  \end{description}
\end{lemma}
\begin{proof}
First, assume that the segment $S$ consists of a single vertex or of
vertical edges only and thus $\MIN(S)=\MAX(S)$, i.e., $\min_{v\in
  S}v_x=\max_{v\in S}v_x$. Thus, $\Ih{e}=\emptyset$ for all $e\in E(S)$,
and there is no $\xi\in \Ih{S}$.  Now, it is easy to see that (O), (O') and
(O'') are always satisfied and, in particular, trivially equivalent. Thus,
assume that $S$ contains at least one non-vertical edge. In this case,
$\min_{v\in S}v_x\neq\max_{v\in S}v_x$ and thus, there is a $\xi\in
\Ih{S}$.  One easily verifies that (O) and (O'') are equivalent.

Next, suppose that $S$ satisfies (O'). Assume, for contradiction, that (O'') is
violated. Since $S$ is connected, for all $\xi\in \Ih{S}\subseteq\mathbb{R}$
there is at least one edge $e$ with $\xi\in \Ih{e}\subseteq \Ih{S}$. Since (O'')
is violated, there are two edges $e = \{u,v\},f=\{a,b\}\in E(S)$ with
$\xi\in \Ih{e}\cap \Ih{f}$.  Now, assume w.l.o.g.\ that $u\prec_S v \preceq_S
a\prec_S b$. By (O'.b) and the fact that $e$ and $f$ are non-vertical,
$u_x< v_x \leq a_x < b_x$ and thus, $\xi < v_x\leq a_x\le \xi$; a
contradiction.
Hence, if $S$ satisfies (O'), then $S$ satisfies (O'') as well.

Now, suppose that $S$ satisfies (O''). As argued above, we can assume that $S$
contains at least one non-vertical edge and thus, $\min_{v\in S} v_x\neq
\max_{v\in S} v_x$. We continue with showing that (O') is satisfied by
induction on the number of edges of $S$. Let $t$ and $t'$ be the terminal
vertices of $S$. As base case, assume that $S$ contains only one edge $e
=\{t,t'\}$. Since $e$ must be non-vertical, we have $\{\min_{v\in
  S}v_x,\max_{v\in S} v_x\} =\{t_x,t'_x\}$, and it is easy to see that (O')
holds.  Assume that (O') is satisfied for all segments $S$ of $P$ having
$m$ edges. Let $S$ be a segment of $P$ that consists of $m+1$ edges. Now,
remove the terminal edge $e=\{u,t\}$ and vertex $t$ from $S$ to obtain the
segment $S_1$ of $P$ with $m$ edges and with terminal vertices $t'$ and
$u$. Similarly, remove the terminal edge $\{u',t'\}$ and vertex $t'$ from
$S$ to obtain the segment $S_2$ of $P$ with $m$ edges and terminal vertices
$t$ and $u'$. By assumption, both $S_1$ and $S_2$ satisfy (O'). We can
assume w.l.o.g.\ that $t'\in\MIN(S_1)$.  In this case, $\preceq_{S_1}$ can
be defined on $V(S_1)$ by choosing $t'$ as the root of $S_1$. Hence, the
other terminal vertex $u$ of $S_1$ is a $\preceq_{S_1}$-maximal element and
thus, by (O'.b), must satisfy $u\in\MAX(S_1)$.  In addition, we have
$t'_x\leq u'_x \leq u_x$, since $t'\preceq_{S_1} u' \preceq_{S_1} u $.

Then, assume for contradiction that $u'\notin\MIN(S_2)$ and, therefore,
$\min_{v\in S_2} v_x \neq u'_x$. Since $S_1$ satisfies (O'.b) and since
$\{t,u\}$ is a terminal edge of $S$, we have $\min_{w\in S_1}w_x = t'_x\leq
u'_x \leq v_x$ for all vertices $v\in V(S_2)\setminus \{t\}\subseteq
V(S_1)$.  This and $\min_{u\in S_2}u_x \neq u'_x$ implies that $t_x<u'_x$
and, in particular, $\min_{v\in S_2} v_x=t_x$. If all edges of $S_1$ are
vertical, it is easy to see that $S$ satisfies (O').  Assume that $S_1$
contains a non-vertical edge. In this case, there is an edge $f=\{a,b\}$ in
$S_1$ such that $a\preceq_{S_1} b$ and $t'_x\leq a_x<b_x\le u_x$. Moreover,
$t_x<u'_x\leq u_x$ implies that there is a $\xi\in \mathbb{R}$ for which
$\min_{v\in f}v_x = a_x\leq \xi <b_x = \max_{v\in f} v_x$ and $\min_{v\in
  e} v_x = t_x \le a_x\leq \xi < b_x \le u_x =\max_{v\in e}v_x$; a
contradiction since $S$ satisfies (O''). Hence, $u'\in\MIN(S_2)$. In this
case, $\min_{v\in S_1} v_x = t'_x\leq u'_x = \min_{u\in S_2} u_x$ implies
that $t'_x\leq v_x$ for all $v\in V(S)$ and thus, $\min_{v\in S}v_x =t'_x$
corresponds to the $x$-coordinates of one of its terminal
vertices. Moreover, since both $S_1$ and $S_2$ satisfy (O'.b) and $t'_x
=\min_{v\in S}v_x$, one easily verifies that (O'.b) must hold for $S$.
This immediately implies that $\max_{v\in S}v_x = t_x$ corresponds to the
$x$-coordinate of the other terminal vertex. Note, however, that $t'_x=t_x$
could be possible.  In either case, $S$ satisfies (O'.a).  Hence, if $S$
satisfies (O), then $S$ satisfies (O') as well. In summary, therefore,
Conditions~(O), (O') and~(O'') are equivalent.
\end{proof}

It is easy to see that Property~(O) is hereditary, that is:
\begin{obs}
  If $S$ is a segment that satisfies (O), then every segment $S'\subseteq
  S$ satisfies (O).
\end{obs}

\begin{definition}
  A segment $S$ is \emph{outstretched} if it satisfies (O) and its terminal
  edges are non-vertical. Moreover, an outstretched segment $S$ of a polygon 
  $P$ is \emph{maximally} outstretched (w.r.t.\ $P$) if there is no outstretched
  segment $S'$ in $P$ with $S\subsetneq S'$.
\end{definition}
Note that outstretched segments $S$ may have empty edge set in which case
$S$ can contain only a single vertex $t=\max(S)=\min(S)$. Moreover,
although terminal edges of $S$ cannot be vertical, $S$ can contain vertical
edges in its interior. In this case, it is even possible that $S$ contains
incident, i.e., consecutive, vertical edges.

\begin{lemma}\label{lem:minmaxS}
  For an outstretched segment $S$ with $E(S)\neq \emptyset$  it always holds
  that $\MIN(S)\cap \MAX(S)=\emptyset$.  In particular, 
  $\MIN(S) = \{t\}$ and $\MAX(S) = \{t'\}$ are singletons consisting
  of the terminal vertices $t$ and $t'$ of $S$ only, i.e., 
  $\min(S)=t$ and $\max(S)=t'$ are well-defined and distinct.
\end{lemma}
\begin{proof}
  By Lemma \ref{lem:O}, $\min_{v\in S}v_x$ and $\max_{v\in S}v_x$
  correspond to the $x$-coordinates of terminal vertices of $S$. Let $t$
  and $t'$ be the terminal vertices and $e=\{u,t\}$ and $f=\{u',t'\}$ be
  the terminal edges of $S$. Note, $e=f$ is possible. Assume w.l.o.g.\ that
  $t\in\MIN(S)$.  Since terminal edges of $S$ are non-vertical, we have
  $\max_{v\in e}v_x\neq \min_{v\in e}v_x$ which implies that $\MIN(S)\neq
  \MAX(S)$ and, in particular, $t'\in \MAX(S)$.  Moreover, by (O'.b) and
  because $e$ and $f$ are non-vertical, it holds that $t_x< u_x\le v_x \leq
  u'_x < t'_x$ for all vertices $v \in V(S)$ with $u\preceq_S v\preceq_{S}
  u'$, where $\preceq_S$ is defined on $V(S)$ by choosing $t$ as the root.
  Thus, $\MIN(S)=\{t\}$ and $\MAX(S)=\{t'\}$, i.e., $\min(S)=t$ and
  $\max(S)=t'$. Moreover, $\MIN(S)\cap\MAX(S)=\emptyset$.
\end{proof}

For later reference, we provide here the following simple result:
\begin{lemma}\label{lem:snf} 
  Let $S,S'$ be two outstretched segments of a polygon $P$ such that
  $E(S)\cap E(S')\ne \emptyset$.  Then, $S^* = S\cup S'$ is an outstretched
  segment of $P$.  
\end{lemma}
\begin{proof} 
  Since $S$ and $S'$ are paths in $P$, $S^* = S\cup S'$ is a connected
  subgraph of $P$ and thus, either $S^*\subsetneq P$ is a path or $S^* = P$
  is a polygon, i.e., a ``graph-theoretical'' cycle. In either case, every
  vertex of $S^*$ must have degree at most two in $S^*$ and $E(S)\cap
  E(S')$ must contain at least one of the terminal edges of $S$ as well as
  of $S'$.  By Lemma \ref{lem:minmaxS}, the terminal vertices
  $t\coloneqq\min(S)$ and $\max(S)$ of $S$ as well as $t'\coloneqq\min(S')$
  and $\max(S')$ are uniquely defined and distinct.
  
  Assume, for contradiction, that $S^*=P$.  In this case, it holds for the
  terminal vertices that $t,t',\max(S),\max(S')\in V(S)\cap V(S')$ since
  $S$ and $S'$ are are paths of a cycle $P$ whose union coincides with
  $P$. But then $t_x = t'_x$ and $\max(S)_x = \max(S')_x$ and thus, $t=t'$
  and $\max(S) = \max(S')$. This and $\MIN(P)\subseteq \MIN(S)\cup\MIN(S')$
  implies that $\min(P) = t$.  Since $S$ and $S'$ are outstretched segments
  and $t=t'$ refers to the unique vertex with minimum $x$-coordinate in $S$
  and $S'$, for the terminal edge $e = \{t,u\}$ of $S$ and $f=\{t,u'\}$ of
  $S'$ it must hold $u\neq u'$; otherwise $t$ would have degree one in
  $S^*=P$; a contradiction to $P$ being a cycle). But then $E(S)\cap E(S')$
  does not contain the terminal edges $e$ and $f$.  By similar arguments,
  $E(S)\cap E(S')$ does not contain the terminal edges $\{\max(S),w\}$ and
  $\{\max(S'),w'\}$.  Hence, $E(S)\cap E(S')$ does not contain any of the
  terminal edges of $S$ and $S'$; a contradiction.

  Therefore, $S^*$ is a path. If $S'\subseteq S$ or $S\subseteq S'$, then
  $S^*=S$ or $S^*=S'$, and thus $S^*$ is an outstretched segment of
  $P$. Otherwise $S'\cap S$ is a path that contain exactly one terminal
  vertex $S$ and $S'$ and $E(S\cap S')\ne\emptyset$ contain at least one
  terminal, and therefore non-vertical edge.  We use the partial orders
  $\preceq_{S}$ and $\preceq_{S'}$ as defined Lemma \ref{lem:O} by using
  $t$ and $t'$ as the root of $S$ and $S'$, respectively.
		
  Assume, for contradiction, that $S^*$ does not satisfy (O). Hence, there
  are two distinct edges $f$ and $g$ in $S^*$ and a $\xi\in \Ih{S^*}$ with
  $\xi\in \Ih{f}\cap \Ih{g}$. Since $S$ and $S'$ satisfy (O), we can
  w.l.o.g.\ assume that $f\in E(S)\setminus E(S')$ and $g\in E(S')\setminus
  E(S)$. Note, both edges $f$ and $g$ are non-vertical since neither $\Ih{f} =
  \emptyset$ nor $\Ih{g} = \emptyset$. Since $e,f\in E(S)$, we have
  $e\preceq_S f$ or $f\preceq_S e$, and we can assume w.l.o.g.\ that
  $e\preceq_S f$. Moreover, since $e,g\in E(S')$, we have either
  $e\preceq_{S'} g$ or $g\prec_{S'}e$. However, $g\prec_{S'}e$ together
  with $e\preceq_S f$, implies $\xi<\max(g)_x\le \max(e)_x\le \min(f)_x\le
  \xi$; a contradiction. Hence, $e\prec_{S'} g$ must hold. Since $S$ and
  $S'$ are paths and $f\in E(S)\setminus E(S')$ and $g\in E(S')\setminus
  E(S)$, there is a unique vertex $w\in V(S)\cap V(S') \subseteq V(S^*)$
  that satisfies the following property: $w$ is adjacent to $v\in
  V(S)\setminus V(S')$ and satisfies $w\preceq_S \max(f)$ and $w$ is
  adjacent to $v'\in V(S')\setminus V(S)$ and satisfies $w\preceq_S
  \max(g)$.  Since $S\cap S'$ is a path with $e\in E(S\cap S')$ and $w\in
  V(S\cap S')$, there is a vertex $v'''\in V(S\cap S')$ adjacent to $w$.
  The latter two observations together imply $\deg_{S^*}(w)\ge 3$; a
  contradiction.

  Hence, $S^*$ satisfies (O). Since, in addition, the terminal edges of
  $S^*$ are terminal edges of $S$ or $S'$ and are, thus, non-vertical, it
  follows that $S^*$ is an outstretched segment.
\end{proof}
   
Note that a similar result as in \Cref{lem:snf} does not hold if we only
claim $V(S)\cap V(S')\ne \emptyset$ instead of $E(S)\cap E(S')\ne
\emptyset$. To see this, consider the two maximal outstretched segments
$S_1$ and $S_2$ in the polygon $P$ as shown in
\Cref{fig:consecutive}(right). Here $S_1$ and $S_2$ intersect in a single
vertex but not in their edges. It is easy to verify that $S_1\cup S_2$ is
not an outstretched segment since Property (O) is violated.

\begin{definition}
  An outstretched segment $S$ of $P$ is maximal if there is no outstretched
  segment $S'$ of $P$ with $S\subsetneq S'$.
\end{definition}
The following properties of maximal outstretched segments will be useful:
\begin{proposition}\label{lem:bSPp} 
  For two distinct maximal outstretched segments $S$ and $S'$ of a polygon
  $P$ the following conditions are satisfied:
  \begin{enumerate}[noitemsep,nolistsep]
  \item $E(S)\cap E(S')=\emptyset$, i.e.\ they are edge-disjoint;
  \item $S\cap S'= \{\min(S),\max(S)\}\cap \{\min(S'),\max(S')\}$
  i.e.\ they can intersect at most in their terminal vertices;
  \item If $S\cap S'\ne \emptyset$, then 
  $\min(S)=\min(S')$ or
    $\max(S)=\max(S')$.
  \item If $|V(S)\cap V(S')|=2$ then $S\cup S'=P$.  
  \end{enumerate}
\end{proposition}
\begin{proof}
  Let $S, S'$ be two distinct maximal outstretched segments of a
  polygon $P$.

  1. If there is an edge $e\in E(S)\cap E(S')$, then \cref{lem:snf} implies
  that $S^*\coloneqq S\cup S'$ is an outstretched segment of $P$ and
  $S,S'\subseteq S^*$. Maximality implies $S=S'=S^*$; a contradiction.

  2. Let $p \in S\cap S'$.  By Statement~(1), $E(S)\cap E(S')=\emptyset$.
  Since $S,S'\subseteq P$ and $P$ is planar we can conclude that $S$ and
  $S'$ can only intersect in a common vertex, i.e., $p\in S\cap S'\subseteq
  V(S)\cap V(S')$.  Now, assume for contradiction that $p \in
  V(S)\setminus\{\min(S),\max(S)\}$ or $p \in
  V(S')\setminus\{\min(S'),\max(S')\}$.  Then, we may assume w.l.o.g.\ $p
  \in V(S)\setminus\{\min(S),\max(S)\}$.  Since $S,S'\subseteq P$ and $P$
  is a graph-theoretical cycle, and since $p\in
  V(S)\setminus\{\min(S),\min(S)\}$, where $\min(S)$ and $\max(S)$ are
  terminal vertices of $S$, we conclude that $\deg_{S}(p)=2=\deg_{P}(p)$.
  This, together with $\deg_{S'}(p)\ge 1$, implies that there is an
  incident edge $\{p,p'\}\in E(S')\subseteq E(P)$ for which we have
  $\{p,p'\}\in E(S)$. In particular, $\{p,p'\}\in E(S)\cap E(S')\ne
  \emptyset$, contradicting Statement~1.

  3. Assume that $S\cap S'\ne \emptyset$
  By Statement~2, $S\cap S'= \{\min(S),\max(S)\}\cap \{\min(S'),\max(S')\}$.
  Moreover, $S,S'\subsetneq S\cup S'$. However, $S\cup S'$ cannot be a
  maximal outstretched segment since, otherwise, we would contradict
  maximality of $S$ and $S'$.  Hence, only $S\cup S' = P$ is
  possible. Since $S$ and $S'$ are edge-disjoint, they must therefore
  intersect in precisely two vertices, i.e.,
  $\{\min(S),\max(S)\} = \{\min(S'),\max(S')\}$. This, together with
  $S\cup S' = P$, implies that $\min(S)=\min(S')$ or $\max(S)=\max(S')$ must
  hold.
  
  4. If $|V(S)\cap V(S')|=2$, then Statement~3 implies that
  $\min(S)=\min(S')$ and $\max(S)=\max(S')$. Since by Statement 1., $S$ and
  $S'$ are edge disjoint and thus internally vertex disjoint in $P$, we
  conclude that $S\cup S'$ is a graph theoretic cycle, and thus coincides
  with $P$.
\end{proof}

\begin{corollary}\label{cor:edge-uniqueS}
  Let $P$ be a polygon. Then, for every non-vertical edge $e\in E(P)$,
  there is a unique maximal outstretched segment $S$ of $P$ such that
  $e\in E(S)$. 
\end{corollary}
\begin{proof}
  Let $e\in E(P)$ be a non-vertical edge. Since $s[e]\subset P$ is an
  outstretched segment of $P$, there is a maximal outstretched segment $S$
  containing $e$. Since maximal outstretched segments are edge-disjoint by
  \Cref{lem:bSPp}~(1), $S$ is uniquely determined.
\end{proof}

\begin{definition}\label{def:consecutive}
  A path (segment) in $P$ that consists of vertical edges only is called
  \emph{vertical} path (segment).  Two maximal outstretched segments $S$
  and $S'$ are \emph{consecutive (along $P$)} if they are distinct and (i)
  $V(S)\cap V(S')\neq \emptyset$ or (ii) there is a vertical path $W$ that
  does not share edges with $S$ and $S'$ and such that $S\cup S'\cup W$ is
  a segment of $P$.
\end{definition} 

\begin{figure}
  \begin{center}
    \includegraphics[width=0.9\textwidth]{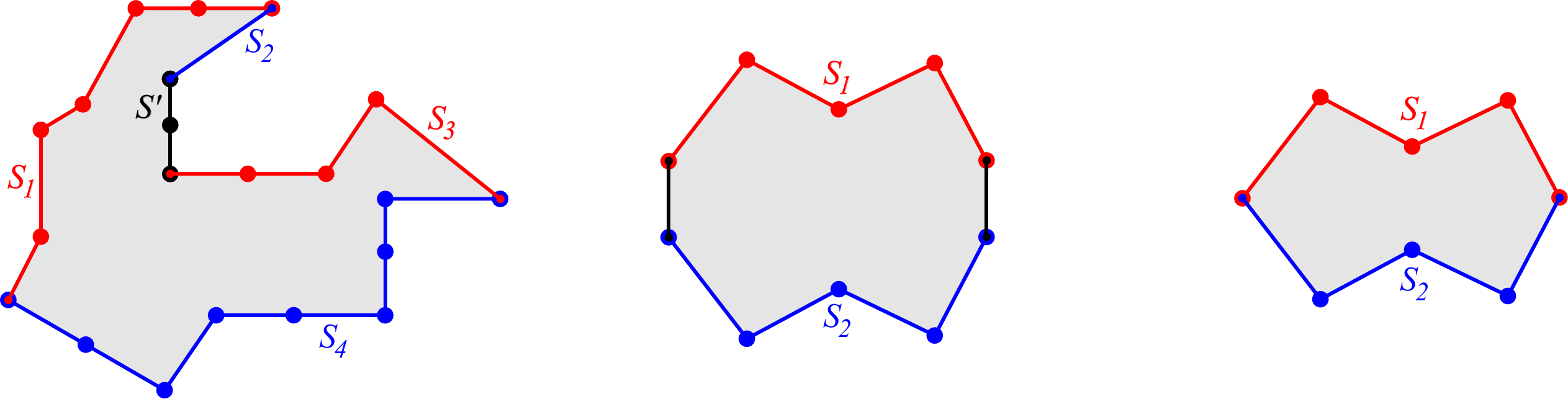}
  \end{center}
  \caption{Decompositions of polygons into segments. In all examples, $S_i$
    is maximal outstretched and $S_i$ and $S_{i+1}$ are consecutive.
    According to \Cref{lem:bSPp}, all maximal outstretched segments $S_i$
    and $S_j$ are edge-disjoint and intersect at most in their terminal
    vertices. If they intersect, $\min(S_j)=\min(S_i)$ or
    $\max(S_j)=\max(S_i)$ must hold, and $S_i$ and $S_j$ are
    consecutive. Otherwise, two consecutive maximal outstretched segments
    $S_i$ and $S_j$ are separated by a single vertical segment and
    $\min(S_j)_x=\min(S_i)_x$ or $\max(S_j)_x=\max(S_i)_x$ must hold. \newline
    \emph{Left:} The vertical segment $S'$ is not contained in a
    maximal outstretched segment. The maximal outstretched segments
    $S_2$ and $S_3$ are consecutive with distinct terminal vertices.
    However, $\min(S_2)_x = \min(S_3)_x$ and the terminal vertices
    $\min(S_2)$ and $\min(S_3)$ are connected by a path consisting of
    vertical edges, namely the vertical segment $S'$. All other
    consecutive maximal outstretched segments intersect in precisely one
    vertex. \newline
    \emph{Middle and Right:} In the middle example the two
    consecutive segments $S_1$ and $S_2$ are vertex-disjoint. In contrast,
    the example on the right shows two consecutive outstretched segments
    $S_1$ and $S_2$ with $|V(S_1)\cap V(S_2)|=2$ and thus, according to
    \Cref{lem:bSPp}, $S_1\cup S_2 = P$.}
  \label{fig:consecutive}
\end{figure}

\Cref{lem:bSPp} implies $|V(S)\cap V(S')|\leq 2$ for two
distinct maximal outstretched segments $S$ and $S'$.  Moreover, if
$|V(S)\cap V(S')| = 1$ then, by \Cref{lem:bSPp}(3), two maximal
outstretched segments $S$ and $S'$ intersect precisely in one of their
terminal vertices.  In this case, it is easy to verify that $S\cup S'$
forms a segment of $P$. Together with \Cref{lem:bSPp}(4), this implies 
\begin{obs}\label{obs:2consecutive}
  For two consecutive maximal outstretched segments $S$ and $S'$ of $P$
  with $V(S)\cap V(S')\neq \emptyset$ it holds that $S\cup S'=P$ or $S\cup
  S'$ is a segment in $P$.
\end{obs}

\begin{lemma}
  \label{lem:consec-x}
  If $S$ and $S'$ are consecutive maximal outstretched segments along $P$,
  then $\min(S)_x=\min(S')_x$ or $\max(S)_x=\max(S')_x$. 
\end{lemma}
\begin{proof}
  In Case (i) of \cref{def:consecutive}, the statement of the lemma follows
  by \Cref{lem:bSPp}~(3). Now, consider Case (ii) of
  \cref{def:consecutive}: Since $W$ is vertical and $S$ and $S'$ are
  consecutive, (at least) one of the following equations must hold:
  $\min(S)_x=\min(S')_x$, $\max(S)_x=\max(S')_x$, $\min(S)_x=\max(S')_x$,
  or $\max(S)_x=\min(S')_x$. In the latter two cases, $S\cup S'\cup W$ is
  an outstretched segment of $P$, contradicting maximality of $S$ and $S'$;
  the remaining two equalities constitute the statement of Lemma.
\end{proof}

In the following observation we collect basic properties of terminal
vertices of consecutive maximal outstretched segments on a simple
polygon. They are immediate consequences of the definition,
\Cref{lem:consec-x} and the fact that maximal outstretched segments
are edge-disjoint graph-theoretic paths on a given simple polygon.

\begin{obs}
  \label{obs:consecutive}
  Let $P$ be a polygon and $S,S'$ distinct maximal outstretched segments of
  $P$ with sets of terminal vertices $\{t_1,t_2\}\in S$ and
  $\{t_1',t_2'\}\in S$. Then, by \cref{def:consecutive}, if
  $\{t_1,t_2\}\cap\{t_1',t_2'\}\ne\emptyset$, then $S$ and $S'$ are
  consecutive along $P$. A terminal vertex $t$ of $S$ is a terminal vertex
  of at most one other maximal outstretched segment (otherwise, $t$ would
  have degree $3$ in $P$ which is not possible).  Moreover, if $t$ is a
  terminal vertex that appears as a terminal vertex in only one maximal
  outstretched segment $S$, then there is a uniquely defined maximal
  outstretched segment $S''$ such that $S$ and $S''$ are consecutive and
  $S''$ has a terminal vertex $t''$ with $t_x=t''_x$. In this case, $t$ and
  $t''$ are the terminal vertices of a vertical segment $W$ since, by
  definition, all non-vertical edges are contained in maximal outstretched
  segments. Moreover, $t''$ does not appear as a terminal vertex in any
  other maximal outstretched segment (otherwise, $t''$ would be contained
  in $W$, $S''$ and some other maximal outstretched segment and thus, would
  have degree $3$ in $P$).
\end{obs}

\section{Sweep-Lines and the Parity of Maximal Outstretched Segments}
\label{sec:sl}

\citet{Bajaj:90} devised a nesting-algorithm based on the ``sweep-line''
approach for overlap-free sets of polygons with the additional condition
that polygons do not touch each other.  Here, we will generalize this
approach to situations where polygons may touch.  Our starting point is the
``ray-casting algorithm'' devised by \citet{Shimrat1962} in 1962 (also
known as ``crossing number algorithm'' or ``even-odd rule algorithm''). It
solves the \emph{point-in-polygon problem} of determining whether a given
point in the plane lies inside, outside, or on the boundary of a polygon,
see \cite{Huang1997} for a survey of the topic.

In the following, let $S$ be an outstretched segment of polygon $P$.  By
\cref{def:prop-O,lem:O}, for every $x\in \Ih{S}$ there is a unique
non-vertical edge $e$ in $S$ such that $x\in \Ih{e}$.  In other words, the
intervals $\Ih{e}$ of the non-vertical edges $e\in E(S)$ partition
$\Ih{S}$.  In upcoming proofs, we need also the identification of edges
with $x \in\Ic{S}=[\min(S)_x,\max(S)_x]$ instead of $\Ih{S}$.  Let $f$
denote the edge in $S$ for which $\max(f)_x = \max(S)_x$.  Since $S$ is
outstretched, $f$ corresponds to one of the two terminal edges of $S$ and
is, in particular, uniquely determined.  Repeating the latter arguments,
for every $x\in \Ic{S}$ there is a unique non-vertical edge $e$ in $S$ such
that $x\in \Ih{e}$ whenever $x\neq \max(S)_x$ and, in case $x = \max(S)_x$,
there is the unique edge $f$ with $x\in \Ih{f}\cup\{\max(S)_x\}$.  In this
way, we obtain a unique identification of elements $x\in \Ic{S}$ and edges
$e$ in $S$. In this case, we say that \emph{the edge $e$ is associated with
$x$}.  Thus, for all $x\in \Ic{S}$, we can define $\yVal{S}{x}$ as the
unique y-coordinate $y^*$ such that $(x,y^*)\in s[u,v]$ for the unique edge
$e = \{u,v\}\in E(S)$ that is associated with $x$.  Note that there might
be additional, vertical edges $e'=\{u',v'\}\in E(S)$ such that $(x,y)\in
s[u',v']$. For these vertical edges, however, $\Ih{e'} = \emptyset$, and
$x\ne t_x$ for a terminal vertex $t$ of a non-vertical edge in $S$.

A direct consequence of the ray-casting algorithm, restated in our
notation, is the following observation:
\begin{proposition} [{\citet[L.\ 2.4]{Bajaj:90}}]
  \label{prop:L24}
  Let $P$ be a polygon and $p=(\xi,y)\in \mathbb{R}^2\setminus P$.  Then,
  $p\in \Int(P)$ if and only if the number of edges $e\in \E(P)$ with
  $\xi\in \Ih{e}$ and $y<\yVal{s[e]}{\xi}$ is odd.
\end{proposition}
\Cref{prop:L24} and \Cref{cor:edge-uniqueS} 
imply
\begin{lemma}\label{lem:inside<->odd}
  Let $P$ be a polygon and $p=(\xi,y)\in \mathbb{R}^2\setminus P$. Then,
  $p\in \Int(P)$ if and only if the number of maximal outstretched segments
  with $\xi\in \Ih{S}$ and $y<\yVal{S}{\xi}$ is odd.
\end{lemma}
\begin{proof}
  The \emph{only-if} direction is a direct consequence of \Cref{prop:L24}
  and \Cref{cor:edge-uniqueS} together with the fact that
  $e\in E(S)$ with $\xi\in \Ih{e} \subseteq \Ih{S}$ implies
  $y<\yVal{s[e]}{\xi} = \yVal{S}{\xi}$.  For the \emph{if} direction
  suppose that the number of maximal outstretched segments $S$ with $\xi\in
  \Ih{S}$ and $y<\yVal{S}{\xi}$ is odd.  Since $\xi\in \Ih{S}$ there must
  be at least one edge $e \in E(S)$ with $x\in \Ih{e}$ and
  $y<\yVal{s[e]}{\xi}$.  By Property (O), for all distinct edges $f,f'\in
  E(S)$ it must hold that $\Ih{f}\cap \Ih{f'} = \emptyset$.  Hence, the
  edge $e$ in $S$ with $x\in \Ih{e}$ and $y<\yVal{s[e]}{\xi}$ is
  unique. Thus, the number of such edges must be odd and \cref{prop:L24}
  implies that $p\in \Int(P)$.
\end{proof}

The aim of this section is to show that instead of considering whether the
number of maximal outstretched segments with $\xi\in \Ih{S}$ and
$y<\yVal{S}{\xi}$ is odd or even for a given $\xi$, one can assign the
property of being ``odd'' or ``even'' to the maximal outstretched segments
$S$ of $P$ themselves, independent of the choice of $\xi$. We start with the
following technical detail:
\begin{lemma} \label{lem:asa} 
  Two maximal outstretched segments $S,S'$ of a polygon $P$ satisfy
  \begin{enumerate}[nolistsep,noitemsep]
  \item $\yVal{S}{\xi}\le\yVal{S'}{\xi}$ for all
    $\xi\in \Ic{S}\cap\Ic{S'}$; or 
  \item $\yVal{S}{\xi}\ge\yVal{S'}{\xi}$ for all
    $\xi\in \Ic{S}\cap\Ic{S'}$.
  \end{enumerate}
\end{lemma}
\begin{proof}
  Assume, for contradiction, that there are two distinct maximal outstretched
  segments $S,S'$ of a polygon $P$ and two $\xi,\xi'\in \Ic{S}\cap \Ic{S'}$
  such that $\yVal{S}{\xi}>\yVal{S'}{\xi}$ and
  $\yVal{S}{\xi'}<\yVal{S'}{\xi'}$.  We assume w.l.o.g.\ that $\xi<\xi'$.
  Note that $f\colon \Ic{S}\to \mathbb{R}$ via $x\mapsto \yVal{S}{x}$ as well
  as $g\colon \Ic{S'}\to \mathbb{R}$ via $x\mapsto \yVal{S'}{x}$ are
  continuous function and we can, therefore, apply the
  intermediate-value-theorem to conclude that there is a $\xi''\in
  (\xi,\xi')$, and thus $\xi''\notin\{\min(S)_x, \max(S)_x\}$, such that 
  $f(\xi'')=g(\xi'')$. But then, we have
  $p\coloneqq(\xi'',f(\xi''))\in S\cap S'$ with 
  $p\notin\{\min(S),\max(S)\}$, 
  which is a contradiction to \cref{lem:bSPp}~(2).
\end{proof}

Let $\mc{S}$ be the set of all maximal outstretched segments of the polygon
$P$.  In the following, we will need to argue about the number of maximal
outstretched segments that are located above a given $S\in\mc{S}$
(including $S$ itself).
\begin{definition}
  Let $\mc{S}$ be the set of all maximal outstretched segments of the
  polygon $P$, $S\in \mc{S}$, and $\xi\in \Io{S}$. Then, we set
  \begin{align*}
    \mc{N}_{\xi,S} &\coloneqq \{S'\in \mc{S} \mid \xi \in 
    \Ih{S'}\text{ and }\yVal{S'}{\xi}\ge \yVal{S}{\xi}\}, \\
    \mc{N}_{\xi,S}^{\min} &\coloneqq \{ S'\in\mc{S} \mid
    \min(S')_x=\xi \text{ and } \min(S')_y\ge \yVal{S}{\xi} \},
    \textnormal{ and} \\
    \mc{N}_{\xi,S}^{\max} & \coloneqq\{ S'\in\mc{S} \mid
    \max(S')_x=\xi \text{ and }\max(S')_y\ge \yVal{S}{\xi}\}.
  \end{align*}
\end{definition}
By definition, we have $\mc{N}_{\xi,S}^{\min}\subseteq \mc{N}_{\xi,S}$ while
$\mc{N}_{\xi,S}^{\max}\not\subseteq \mc{N}_{\xi,S}$ is possible. We have
excluded $\xi=\min(S)_x$ and $\xi=\max(S)_x$ since the outstretched
segments preceding and succeeding $S$ along $P$, respectively, would always
be included irrespective of whether they are all above or below $S$ for
other values of $\xi$.

\begin{lemma}\label{lem:emm} 
  Let $S$ a maximal outstretched segment of a polygon $P$ and 
  $\xi\in \Io{S}$.
  Then, both $|\mc{N}_{\xi,S}^{\min}|$
  and $|\mc{N}_{\xi,S}^{\max}|$ are even.  
\end{lemma}
\begin{proof}
  Let $P$ be a polygon, $S$ a maximal outstretched segment of $P$, let
  $\xi\in \Io{S}$, and set $N^{\min}\coloneqq
  |\mc{N}_{\xi,S}^{\min}|$ and $N^{\max} \coloneqq
  |\mc{N}_{\xi,S}^{\max}|$, respectively.  We determine $N^{\min}$ and
  $N^{\max}$ by traversing terminal vertices $t$ of maximal outstretched
  segments in the polygon $P$ in clockwise order starting from $\min(S)$.
  Since $\xi=\min(S)_x$ is excluded by assumption, we still have
  $N^{\min}=N^{\max}=0$ after departing from $\min(S)$. Note that
  $N^{\min}$ and $N^{\max}$ remain unchanged whenever $t_y<\yVal{S}{\xi}$
  or $t_x\ne \xi$.  Hence, suppose in the following that $t$ is a terminal
  vertex of $S$ and satisfies $t_x=\xi$ and $t_y\geq \yVal{S}{\xi}$.  If
  $t$ is shared by two consecutive maximal outstretched segments $S$ and
  $S'$, then by \Cref{lem:consec-x} either $t_x = \min(S)_x = \min(S')_x$
  or $t_x = \max(S)_x = \max(S')_x$, and thus either $N^{\max}$ or
  $N^{\min}$ is increased by $2$.  Assume now that $t$ is a terminal vertex
  that is not shared by consecutive segments. Let $t'$ be the first
  terminal vertex of the maximal outstretched segment $S'$ that is
  encountered after $t$ when traversing the terminal vertices in clockwise
  order.  Then $t$ and $t'$ are connected by a path $W$ consisting of
  vertical edges only, since every non-vertical edge is contained in a
  maximal outstretched segment. According to \cref{obs:consecutive}, $S$
  and $S'$ are consecutive, and $t'$ satisfies $t'_x=t_x=\xi$. Since the
  polygon $P$ is simple by assumption, no other terminal vertex lies on the
  vertical segment $W$. Furthermore, if $t_y> \yVal{S}{\xi}$ then $t'_y>
  \yVal{S}{\xi}$, since otherwise $W$ and $S$ would intersect or touch,
  contradicting that $P$ is simple.  Again, by \Cref{lem:consec-x}, either
  both $t$ and $t'$ are the minima or both are the maxima of their
  respective segments. The consecutive pair $\{t,t'\}$ therefore also adds
  $2$ to either $N^{\max}$ and $N^{\min}$.  Thus, both $N^{\max}$ and
  $N^{\min}$ are even when the traversal of $P$ returns to $\min(S)$.
\end{proof}

\begin{lemma} \label{lem:mc-N-investi}
  Let $S$ be a maximal outstretched segment of a polygon $P$ and $v\in
  V(P)$.  Suppose that there are $\xi,\xi'\in \Io{S}$ such that $x\coloneqq
  v_x\in (\xi,\xi']$ but $w_x\notin (\xi,\xi']$ for all $w \in V(P)$ with
  $w_x\neq v_x$.  Then,
  $|\mc{N}_{\xi,S}|-|\mc{N}_{x,S}^{\max}|=
  |\mc{N}_{\xi',S}|-|\mc{N}_{x,S}^{\min}|$.
\end{lemma}
\begin{proof}
  Let $S$ be a maximal outstretched segment of a polygon $P$ and $v\in 
  V(P)$.   Moreover, suppose that there are $\xi,\xi'\in \Io{S}$ such that
  $x\coloneqq v_x\in (\xi,\xi']$ but $w_x\notin (\xi,\xi']$ for  
  all $w \in V(P)$ with $w_x\neq v_x$. In particular, the latter implies that
  there is no $w\in V(P)$ with $w_x\in (\xi,\xi']\setminus\{x\}$. 
  By assumption, $x,\xi,\xi'\in \Ih{S}$.
  Since $S$ is fixed, we simplify the notation, and write 
  $\mc{N}_{\xi}\coloneqq \mc{N}_{\xi,S}$, 
  $\mc{N}_{\xi'}\coloneqq \mc{N}_{\xi',S}$,
  $\mc{N}_{x}^{\max}\coloneqq\mc{N}_{x,S}^{\max}$ and 
  $\mc{N}_{x}^{\min}\coloneqq\mc{N}_{x,S}^{\min}$.  
  First, we prove \emph{Claim~1:} 
  \emph{(a)} $\mc{N}_{\xi}\setminus 
  \mc{N}_{x}^{\max} \subseteq 
  \mc{N}_{\xi'}\setminus \mc{N}_{x}^{\min}$,\ \
  \emph{(b)} $\mc{N}_{\xi'}\setminus \mc{N}_{x}^{\min} \subseteq 
  \mc{N}_{\xi}\setminus \mc{N}_{x}^{\max}$\ \ and\ \
  \emph{(c)} $|\mc{N}_{\xi} \setminus 
  \mc{N}_{x}^{\max}|=|\mc{N}_{\xi'}\setminus \mc{N}_{x}^{\min}|$.
  
  First, consider Claim~(1a), and let $S'\in \mc{N}_{\xi}\setminus
  \mc{N}_{x}^{\max}$.  Hence, $\min(S')_x\le \xi < \max(S')_x$ and
  $\yVal{S'}{\xi}\ge \yVal{S}{\xi}$.  Thus, if $\max(S')_x=x$, then
  \Cref{lem:asa} implies $\yVal{S'}{x}\ge \yVal{S}{x}$, and thus, $S'\in
  \mc{N}_{x}^{\max}$; a contradiction.  Hence, $\min(S')_x\le \xi <
  \max(S')_x\ne x$.  This, together with $\max(S')\in V(P)$ and the fact
  that there is no $w\in V(P)$ with $w_x\in (\xi,\xi']\setminus\{x\}$,
  implies $\min(S')\le \xi<\xi'<\max(S')_x$, and thus, $\xi'\in
  \Ih{S'}\cap \Ih{S}$.  Hence, $\yVal{S'}{\xi}\ge \yVal{S}{\xi}$,
  together with \Cref{lem:asa}, implies $\yVal{S'}{\xi'}\ge
  \yVal{S}{\xi'}$, and thus, $S'\in \mc{N}_{\xi'}$.  Moreover,
  $\min(S')_x\le \xi<x$ implies $\min(S')_x\ne x$. 
  Consequently,
  $S'\notin\mc{N}_{x}^{\min}$.  The latter two observations together
  imply Claim~(1a).  By similar reasoning, Claim~(1b) is satisfied.
  Furthermore, Claim~(1a) and~(1b) imply 
  $\mc{N}_{\xi} \setminus \mc{N}_{x}^{\max}=
  \mc{N}_{\xi'}\setminus \mc{N}_{x}^{\min}$, 
  and thus, Claim~(1c) holds.

  We continue with showing 
  \textit{{Claim~2:} $\mc{N}_{x}^{\max} \subseteq \mc{N}_{\xi}$.}
  \noindent
  Let $S'\in \mc{N}_{x}^{\max}$. Thus, $\max(S')_x=x$. Since 
  $\min(S')_x<\max(S')_x$, we have $\min(S')_x\ne x$.
  Thus, $x\in (\xi,\xi']$ implies 
  $\min(S')_x\le\xi<x=\max(S')_x$,
  and thus, $\xi\in \Ih{S'}\cap \Ih{S}$.  Hence,
  $\yVal{S'}{x}\ge \yVal{S}{x}$, together with \Cref{lem:asa}, implies
  $\yVal{S'}{\xi}\ge \yVal{S}{\xi}$, and thus, $S'\in \mc{N}_{\xi}$.
  Hence, Claim~(2) holds.
  By similar arguments one can prove 
  \textit{{Claim~3:} $\mc{N}_{x}^{\min} \subseteq \mc{N}_{\xi'}$.}

  Finally, Claim (2) and (3) imply that $|\mc{N}_{\xi} \setminus 
  \mc{N}_{x}^{\max}|=|\mc{N}_{\xi}|-|\mc{N}_{x}^{\max}|$
  and $|\mc{N}_{\xi'}\setminus 
  \mc{N}_{x}^{\min}|=|\mc{N}_{\xi'}|-|\mc{N}_{x}^{\min}|$, respectively.
  This, together with Claim~(1c), implies 
  $|\mc{N}_{\xi}|-|\mc{N}_{x}^{\max}|=|\mc{N}_{\xi'}|-|\mc{N}_{x}^{\min}|$.
\end{proof}

\begin{lemma}\label{lem:eol} 
  Let $S$ be a maximal outstretched segment of a polygon $P$.  Then,
  $|\mc{N}_{\xi,S}|$ is either always odd or always even independent of the
  choice of $\xi\in \Io{S}$.
\end{lemma}
\begin{proof}
  Let $S$ be a maximal outstretched segment of a polygon $P$, and 
  let $X'\coloneqq \{ v_x\in \Ic{S}\mid v\in V(P)\}$.
  In particular, we assume $X'=\{x\lc{1},x\lc{2},x\lc{3},\ldots,x_m\}$ 
  to be sorted such that, for all $i,j\in
  \{1,\ldots,m\}$, we have $i\le j$ if and only if $x_i\le x_j$. This and
  \Cref{lem:minmaxS} implies that $x_1=\min(S)_x \neq \max(S)_x=x_m$, and
  thus, $m\ge 2$.  Since $S$ is fixed, we simply the notation and write
  $\mc{N}_\xi\coloneqq \mc{N}_{\xi,S}$, $\mc{N}^{\min}_\xi \coloneqq
  \mc{N}_{\xi,S}^{\min}$ and $\mc{N}^{\max}_\xi \coloneqq
  \mc{N}_{\xi,S}^{\max}$, where $\xi\in \Io{S}$.
  
  First, assume that $|X'|=2$, i.e., $X'=\{\min(S)_x,\max(S)_x\}$.
  Thus, there is no $v\in V(P)$ with $\min(S)_x<v_x<\max(S)_x$, and thus,
  there is no $S'\in \mc{S}$ with
  $\min(S)_x < \min(S')_x, \max(S')_x < \max(S)_x$.  Let
  $\xi,\xi'\in \Io{S}$  and $S'\in \mc{N}_\xi$.  Since $S'\in \mc{N}_\xi$,
  we have $\min(S')_x\le\xi<\max(S')_x$.  Applying the previous
  arguments again, $\min(S')_x\le \min(S)_x<\xi<\max(S)_x\le
  \max(S')_x$.  Thus, $\xi'\in \Io{S}\subseteq \Ih{S'}$.
  Moreover, $S'\in \mc{N}_\xi$ implies $\yVal{S'}{\xi}\ge
  \yVal{S}{\xi}$ and thus, \Cref{lem:asa} yields $\yVal{S'}{\xi'}\ge
  \yVal{S}{\xi'}$, and hence $S'\in \mc{N}_{\xi'}$. Therefore, we
  have $\mc{N}_{\xi}\subseteq \mc{N}_{\xi'}$.  Interchanging the role of
  $\xi$ and $\xi'$ shows $\mc{N}_{\xi'}\subseteq \mc{N}_{\xi}$, and thus
  $\mc{N}_{\xi} = \mc{N}_{\xi'}$. Hence, $|\mc{N}_\xi|$ is either always
  odd or always even, independent of the choice of 
  $\xi\in\Io{S}$.

  Now, assume that $|X'|=m\ge 3$ and fix some $i\in
  \{2,\ldots,m-1\}$. Note that $X'$ is a set consisting of pairwise
  distinct elements. Since $V(P)$ is finite, 
  we can choose $\xi,\xi'\in \mathbb{R}$ such that $x_{i-1}<\xi<x_i\le 
  \xi'<x_{i+1}$ and such that  $w_x\notin (\xi,\xi']$ for all $w\in V(P)$ with 
  $w_x\ne x_i$. By definition of $X'$, there is at least one vertex $v\in V(P)$
  with $v_x=x_i$.
  The latter two arguments together with \Cref{lem:mc-N-investi} imply that
  $|\mc{N}_{\xi}|-|\mc{N}_{x_i}^{\max}|=|\mc{N}_{\xi'}|-|\mc{N}_{x_i}^{\min}|$.
  Since $|\mc{N}_{x_i}^{\min}|$ and $|\mc{N}_{x_i}^{\max}|$ are even by
  \Cref{lem:emm}, it follows that $|\mc{N}_{\xi}|$ is even if and only if
  $|\mc{N}_{\xi'}|$ is even.  Therefore, we conclude that $|\mc{N}_{\xi}|$
  is either even for all $\xi\in \Io{S}$, or $|\mc{N}_{\xi}|$
  is odd for all $\xi$ in this interval.
\end{proof}

\begin{figure}[t]
  \begin{minipage}{0.6\textwidth}
    \caption{Relationships between the parity of maximal outstretched
      segments and location of points.  The Polygon $P$, which bounded by
      the gray area, contains four maximal outstretched segments $S_1$,
      $S_2$, $S_3$, and $S_4$. Three sweep-lines are indicated at positions
      $\xi$, $\xi'$, and $\xi''$. We observe $\xi,\xi',\xi''\in\Ih{S_4}$
      and $\mc{N}_{\xi,S_4} = \mc{N}_{\xi',S_4} = \{S_1,S_4\} \neq
      \{S_1,S_2,S_3,S_4\}=\mc{N}_{\xi'',S_4}$. Nevertheless, the the
      cardinality of these sets is always even, and thus we have
      $\varpi(S_4)=0$.  In particular, the blue segments $S_2$ and $S_4$
      have the parity $\varpi(S_2)=\varpi(S_4)=0$, while the red segments
      $S_1$ and $S_3$ have the parity $\varpi(S_1)=\varpi(S_3)=1$.  Note
      that parities of consecutive segments necessarily alternate.\newline
      For the three points $p\in \Int(P)$ and $p',p''\in\Ext(P)$, we
      observe the following: (1)~$p_x\in \Ih{S_1}$ and
      $p_y<\yVal{S_1}{p_x}$, (2)~there is no maximal outstretched segment
      $S_i$ of $P$ with $p_x\in \Ih{S_i}$ and
      $p_y<\yVal{S_i}{p_x}<\yVal{S_3}{p_x}$, and (3)~$\varpi(S_1)=1$, see
      \cref{prop:ibp}.  No such segment exists for $p'\notin \Int(P)$ and
      $p''\notin \Int(P)$, respectively. }
  \label{fig:pari}
  \end{minipage} \hfil \begin{minipage}{0.3\textwidth}
    \includegraphics[width=1\textwidth]{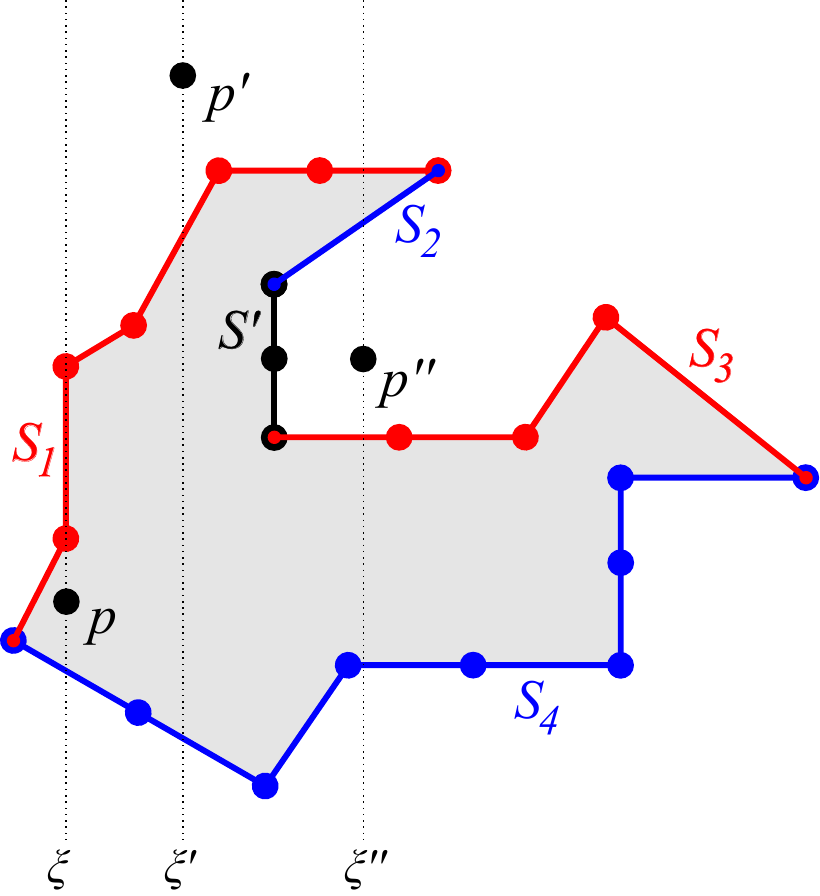}
  \end{minipage}
\end{figure}

\begin{definition} \label{def:parity}
  Let $S$ be a maximal outstretched segment of a polygon $P$.  Then, the
  \emph{parity $\varpi$ of $S$} is
  \begin{equation*}
    \varpi(S)\coloneqq \begin{cases}
      0, \text{\ if } |\mc{N}_{\xi,S}| \text{ is even for all }	
      \xi\in \Io{S}\\
      1, \text{\ otherwise, i.e., }|\mc{N}_{\xi,S}| \text{ is odd for all }
      \xi\in \Io{S}\\
    \end{cases}
  \end{equation*}
\end{definition}
Note that, by \Cref{lem:eol}, $|\mc{N}_{\xi,S}|$ is either always even or
always odd for any $\xi\in\Io{S}$. Therefore, the parity
$\varpi(S)\in \{0,1\}$ is well-defined for every maximal outstretched
segment $S$ of $P$. Hence, the choice of $\xi\in \Io{S}$ in
the definition is arbitrary. 
Intuitively, the parity $\varpi(S)$ determines whether $\Int(P)$ is above
or below the outstretched segment $S$: if $\varpi(S)=0$, then $\Int(P)$ is
``above'' the segment $S$, and if $\varpi(S)=1$, then $\Int(P)$ is
``below'' the segment $S$. Therefore, for a given point $p$, it suffices to
determine the parity $\varpi(S)$ of the segment ``immediately above'' $p$
to determine whether or not $p\in\Int(P)$, see \cref{fig:pari}.
The next result, which is a
simple consequence of \Cref{lem:inside<->odd} and the definition of the
parity function $\varpi$, states this observation more formally:
\begin{proposition} \label{prop:ibp} 
  Let $P$ be a polygon and $p=(\xi,y)\in \mathbb{R}^2\setminus P$.  Then,
  $p\in \Int(P)$ if and only if there is a maximal outstretched segment $S$
  of $P$ such that
  \begin{enumerate}[noitemsep,nolistsep]
  \item $\xi\in \Ih{S}$ and $y<\yVal{S}{\xi}$,
  \item there is no maximal outstretched segment $S'$ of $P$ with $\xi\in
    \Ih{S'}$ such that that $y< \yVal{S'}{\xi}< \yVal{S}{\xi}$,
  \item $\varpi(S)=1$.
  \end{enumerate} 
\end{proposition}
\begin{proof}
  Let $P$ be a polygon and $p=(\xi,y)\in \mathbb{R}^2\setminus P$.  By
  \Cref{lem:inside<->odd}, $p\in \Int(P)$ if and only if the number of
  maximal outstretched segments $S$ of $P$ that satisfies Statement~(1) is
  odd.  In particular, if $p\in \Int(P)$, then there is a maximal
  outstretched segment $S$ of $P$ satisfying Statement~(1).  Moreover, we
  may choose the maximal outstretched segment $S$ of $P$ that satisfies
  Statements~(1) and~(2). Thus, is remains to show that, in the presence of
  Statements~(1) and~(2), $p\in\Int(P)$ and $\varpi(S)=1$ are equivalent.
  By assumption, $p$ is not contained in the boundary of $P$, i.e.,
  {$p\notin P$}. Since $\Int(P)$ is an open set, $p\in\Int(P)$ if and only
  if $p'\in\Int(P)$ for all $p'\in\mathbb{R}^2$ that satisfy $\Vert
  p-p'\Vert<\delta$ for sufficiently small $\delta>0$. In particular, we
  have $p=(\xi,y)\in \Int(P)$ if and only if $p'\coloneqq(\xi+\delta,y)\in
  \Int(P)$. In this case, we have $p'\in \mathbb{R}^2\setminus P$.  Hence,
  by assumption, there is a maximal outstretched segment $S'$ of $P$ that
  satisfies Statement~(1) and~(2) for $(\xi+\delta, y)=p'\in
  \mathbb{R}^2\setminus P$.  Since $\delta>0$ was sufficiently small, we
  can assume that $\min(S')_x\le \xi<\xi+\delta<\min(S')_x$, i.e.,
  $\xi+\delta\in (\min(S')_x,\max(S')_x)$.  Then, by definition of
  $\varpi(S)$, the number of segments that satisfy Statement~(1) is odd
  (i.e., $p'\in\Int(P)$) if and only if $\varpi(S)=1$ (i.e., Statement~(3)
  holds), see \Cref{lem:inside<->odd}.  This, together with $p \in \Int(P)$
  iff $p'\in \Int(P)$, implies $p\in\Int(P)$ if and only if Statement~(3)
  holds.
\end{proof}

Inverting the direction of the $y$-axis immediately implies the following
``mirror image'' of \Cref{prop:ibp}:
\begin{corollary} \label{cor:ibp-} 
  Let $P$ be a polygon and $p=(\xi,y)\in \mathbb{R}^2\setminus P$.  Then,
  $p\in \Int(P)$ if and only if there is a maximal outstretched segment $S$
  of $P$ such that
  \begin{enumerate}[noitemsep,nolistsep]
  \item $\xi\in \Ih{S}$ and $y>\yVal{S}{\xi}$,
  \item there is no maximal outstretched segment $S'$ of $P$ with $\xi\in
    \Ih{S'}$ such that $y> \yVal{S'}{\xi}> \yVal{S}{\xi}$,
  \item $\varpi(S)=0$.
  \end{enumerate} 
\end{corollary}

The following simple consequence of \Cref{prop:L24} and the definition of
consecutive maximal outstretched segments will be useful later on.
\begin{corollary} \label{cor:maxpar=1}
  Let $P$ be a polygon, let $\mc{S}$ be the set of all maximal outstretched
  segments of $P$, and let $S\in \mc{S}$.  If there is a $\xi\in \Io{S}$
  such that $\yVal{S}{\xi}\ge \yVal{S'}{\xi}$ for each $S'\in \mc{S}$ with
  $\xi\in \Ih{S'}$, then $\varpi(S)=1$.  Moreover, if $S$ and $S'\in
  \mc{S}$ are consecutive, then $\varpi(S)\ne\varpi(S')$.
\end{corollary}
\begin{proof}
  Let $P$ be a polygon, let $\mc{S}$ be the set of all maximal outstretched
  segments of $P$, and let $S\in \mc{S}$.  If there is a $\xi\in \Io{S}$
  such that $\yVal{S}{\xi}\ge \yVal{S'}{\xi}$ for each $S'\in \mc{S}$ with
  $\xi\in \Ih{S'}$, then $\mc{N}_{\xi,S}=\{S\}$ and, therefore,
  $\varpi(S)=1$.  Moreover, if $S$ and $S'$ are consecutive, then, by
  \Cref{lem:consec-x}, $\min(S)_x=\min(S)_x$ or $\max(S)_x=\max(S)_x$. In
  the first case, we assume w.l.o.g.\ that $\yVal{S}{\xi}<\yVal{S'}{\xi}$
  for $\xi\coloneqq\min(S)_x+\delta$ with $\delta>0$ sufficiently
  small. Since $P$ is simple, no other segment $S''\in \mc{S}$ contains a
  point $(\xi,\yVal{S''}{\xi})$ with
  $\yVal{S}{\xi}\le\yVal{S''}{\xi}\le\yVal{S'}{\xi}$.  Hence, since
  $\xi\notin \{\max(S)_x,\max(S')_x\}$, we conclude
  $|\mc{N}_{\xi,S'}|-|\mc{N}_{\xi,S}|=1$. Thus, we can assume
  w.l.o.g.\ that $|\mc{N}_{\xi,S'}|$ is odd and $|\mc{N}_{\xi,S}|$ is even.
  Hence, we can apply \Cref{lem:eol} and conclude that $|\mc{N}_{\xi',S'}|$
  is odd for all $\xi'\in \Io{S}$, and $|\mc{N}_{\xi',S}|$ is even for all
  $\xi'\in \Io{S}$.  By definition, $\varpi(S)\ne\varpi(S')$.  The same
  conclusion follows by an analogous argument for $\max(S)_x=\max(S')_x$.
\end{proof}
 
\section{Investigating the Set of Maximally Outstretched Segments}
\label{sec:invest}

So far, we have considered a single polygon.  From here on, we focus on a
set $\P$ of polygons and the corresponding set $\Int(\P) = \{\Int(P)\mid
P\in \P\}$. Since $\Ih{P}$ is the
projection of $P$ onto the $x$-axis, we observe the following:
\begin{obs} \label{obs:I-P-sub}
  Let $\P$ be a set of polygons, and $P,P' \in \P$.  Then,
  $\Int(P)\subseteq \Int(P')$ implies $\Ih{P}\subseteq \Ih{P'}$.
  Conversely, $\Ih{P}\cap \Ih{P'}=\emptyset$ implies
  $\Int(P)\cap\Int(P')=\emptyset$.
\end{obs}
Note that, in general, neither $\Ih{P}\subseteq \Ih{P'}$ implies
$\Int(P)\subseteq \Int(P')$ nor $\Int(P)\cap\Int(P')=\emptyset$ implies
$\Ih{P}\cap \Ih{P'}=\emptyset$; see \cref{fig:non-sweep} for an
example.
Moreover, we will make use of the following three simple result
regarding the boundaries of overlap-free polygons, see
\cref{fig:overlapping} for an illustrative example:
\begin{lemma} \label{lem:s-o-r} 
  Let $P$ and $P'$ be two overlap-free polygons.  Then,
  $P\cap\Int(P')\ne\emptyset$ implies $\Int(P)\subsetneq \Int(P')$ and
  $P\subseteq \Int(P')\cup P'$.  Moreover, $\Ext(P')\cap P\ne\emptyset$
  implies $\Int(P)\nsubseteq \Int(P')$.  In particular, there are no
  $p,p'\in P$ such that $p\in \Int(P')$ and $p'\in \Ext(P')$.
\end{lemma}
\begin{proof}
  Let $P$ and $P'$ be two overlap-free polygons. First, assume that $p\in
  P\cap\Int(P')$. Hence, $p\in \Int(P')\setminus \Int(P)\ne
  \emptyset$. Since $p\in P$ is on the ``boundary'' of $P$ and since $p\in
  \Int(P')$ is \emph{not} on the ``boundary'' of $P'$, for every
  $\varepsilon>0$, there is a $p'\in \mathbb{R}^2$ with $||p-p'|| \le
  \varepsilon$ such that $p'\in \Int(P)\cap \Int(P')$. Hence, $\Int(P)\cap
  \Int(P')\ne\emptyset$.  This, together with $\Int(P')\setminus \Int(P)\ne
  \emptyset$ and the assumption that $P$ and $P'$ are overlap-free implies
  $\Int(P)\subsetneq \Int(P')$.  Now, assume that $p\in P\cap\Ext(P')$. By
  similar reasoning, there is a $p' \in \mathbb{R}^2$ such that $p' \in
  \Int(P)\cap\Ext(P')$.  This together with the assumption that $P$ and
  $P'$ are overlap-free implies $\Int(P)\nsubseteq\Int(P')$.  In summary,
  $P\cap\Int(P')\ne\emptyset$ and $P\cap\Ext(P')\ne\emptyset$ can never
  occur.
\end{proof}

\begin{lemma} \label{lem:area}
  Let $P$ and $P'$ be two overlap-free polygons
  such that $\Int(P)\cap\Int(P')\ne\emptyset$.  Then,
  $\Int(P)\subsetneq\Int(P')$ if and only if $\area{P}<\area{P'}$. In
  particular, $P=P'$ if and only if $\area{P}=\area{P'}$.
\end{lemma}
\begin{proof}
  Since polygons are overlap-free, $\Int(P)\cap\Int(P')\ne\emptyset$
  implies $P=P'$, $\Int(P)\subsetneq\Int(P')$ or
  $\Int(P')\subsetneq\Int(P)$. Clearly, $\Int(P)\subseteq \Int(P')$ implies
  $\area{P}\le\area{P'}$. Since two polygons are different only if their
  vertex sets differ, $P\ne P'$ and $\Int(P)\subseteq\Int(P')$ implies that
  there is vertex $v\in V(P)$ such that $v\in\Int(P')$, and thus
  $\area{P}<\area{P'}$. Thus, $\Int(P)\subseteq\Int(P')$ and
  $\area{P}=\area{P'}$ implies $P=P'$.
\end{proof}

\begin{figure}
  \begin{minipage}{0.55\textwidth}
    \caption{A ``crossing'' pair of maximal outstretched segment $S_1$ and
      $S_2$ implies that the corresponding polygons $P_1$ and $P_2$
      overlap. In particular, there are two positions $\xi,\xi' \in
      \Ih{S_1}\cap\Ih{S_2}$ such that $\yVal{S_1}{\xi}>\yVal{S_2}{\xi}$ and
      $\yVal{S_1}{\xi'}<\yVal{S_2}{\xi'}$.  Thus, $S_1$ and $S_2$ intersect
      (illustrated by the dashed line), and \cref{lem:mos-le-1} implies
      that $P_1$ and $P_2$ are overlapping.}
    \label{fig:overlapping}
  \end{minipage} \hfil \begin{minipage}{0.35\textwidth}
    \includegraphics[width=1\textwidth]{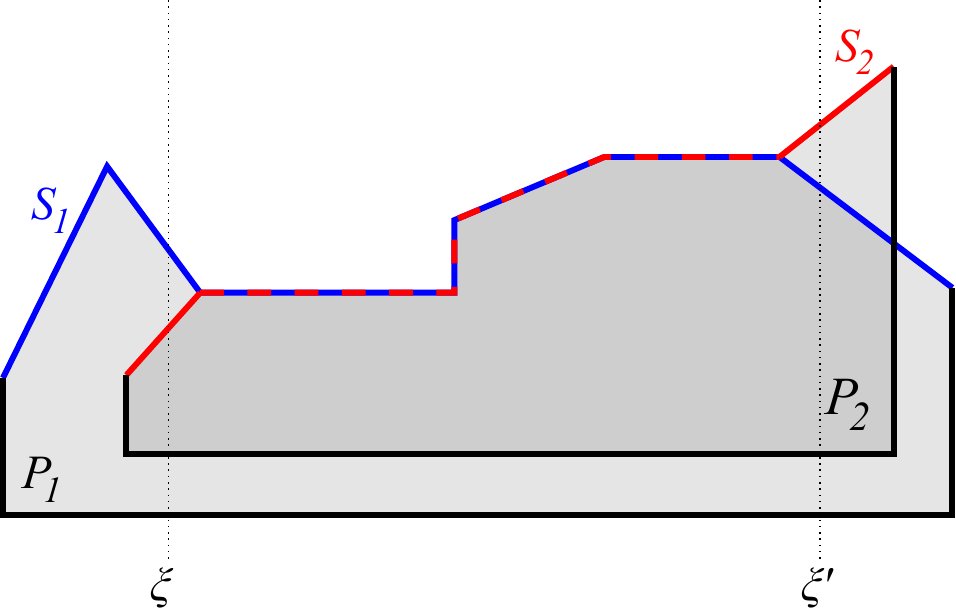}
  \end{minipage}
\end{figure}

\begin{lemma}\label{lem:mos-le-1}
  Let $\P$ be a set of overlap-free polygons and $P,P'\in \P$. 
  Moreover, let $S$ and $S'$ be maximal
  outstretched segments of $P$ and $P'$, respectively.  Then, 
  we have $\yVal{S}{\xi}\le\yVal{S'}{\xi}$ for all 
  $\xi\in \Ic{S}\cap \Ic{S'}$, or
  $\yVal{S'}{\xi}\le\yVal{S}{\xi}$ for all 
  $\xi\in \Ic{S}\cap \Ic{S'}$.
\end{lemma}
\begin{proof}
  Let $(S,P),(S',P') \in \SP{\P}$.  If $P=P'$, the statement follows from
  \cref{lem:asa}.  Hence, suppose that $P\neq P'$.  Now, assume for
  contradiction that there are two $\xi_1,\xi_2\in \Ic{S}\cap \Ic{S'}$ such
  that $\yVal{S'}{\xi_1}>\yVal{S}{\xi_1}$ and
  $\yVal{S'}{\xi_2}<\yVal{S}{\xi_2}$.
  We may assume w.l.o.g.\ that $\xi_1<\xi_2$.  By similar arguments as in
  the proof of \Cref{lem:asa}, the intermediate-value-theorem implies that
  the intersection $\widehat{S^*}\coloneqq \widehat{S}\cap\widehat{S'}$
  with $\widehat{S}\coloneqq \{ p\in S\mid \xi_1\le p_x\le \xi_2 \}$ and
  $\widehat{S'}\coloneqq \{ p\in S'\mid \xi_1\le p_x\le \xi_2 \}$ is
  non-empty.  In particular, it is easy to verify that $\xi_1$ and $\xi_2$
  can be chosen such that the intersection $\widehat{S^*}$ is connected.
  Note that $\widehat{S^*}$ is not necessarily an isolated point and may
  involve vertical edges.  Moreover, $\widehat{S'}\setminus \widehat{S^*}$
  decomposes into exactly two segments $\widehat{S'_1}$ and
  $\widehat{S'_2}$, where w.l.o.g.\ $p_1\coloneqq
  (\xi_1,\yVal{S}{\xi_1})\in \widehat{S'_1}$ and $p_2\coloneqq
  (\xi_2,\yVal{S}{\xi_2})\in \widehat{S'_2}$.  Moreover,
  \cref{lem:bSPp}~(2) implies that any sufficiently small open neighborhood
  $U$ of $\widehat{S^*}$ does not intersect with other maximal outstretched
  segment of $P$ and $P'$, i.e., for a sufficiently small $\delta>0$, we
  have $U\coloneqq\{p\in \mathbb{R}^2\mid \Vert p-p'\Vert < \delta
  \textnormal{ for some }p'\in \widehat{S^*}\}$.  Since
  $\widehat{S'}=\widehat{S'_1}\cupdot \widehat{S^*} \cupdot
  \widehat{S'_2}$, it is straightforward to verify that there is a $p'_1\in
  \widehat{S'_1}\cap U$ and $p'_2\in \widehat{S'_2}\cap U$.  In particular,
  $p'_1,p'_2\notin S$, and thus, $(p'_1)_y\ne \yVal{S}{(p'_1)_x}$ and
  $(p'_2)_y\ne \yVal{S}{(p'_2)_x}$.  Moreover, $p'_1,p'_2\notin S$,
  together with $p'_1,p'_2 \in U$, implies $p'_1,p'_2\notin P$.
  
  Now, assume for contradiction that $(p'_1)_y<\yVal{S}{(p'_1)_x}$, and
  thus, $(p'_1)_y<\yVal{S}{(p'_1)_x}=\yVal{\widehat{S}}{(p'_1)_x}$.  This,
  together with $p'_1\in \widehat{S'_1}$ and
  $\yVal{\widehat{S'_1}}{\xi_1}=\yVal{S'}{\xi_1}>\yVal{S}{\xi_1}=
  \yVal{\widehat{S}}{\xi_1}$, implies $\widehat{S'_1}\cap
  \widehat{S}\ne\emptyset$.  In particular, there is a $p \in
  \widehat{S}\cap\widehat{S'_1} \subseteq \widehat{S}\cap \widehat{S'}$
  with $p\notin \widehat{S^*}$; a contradiction to
  $\widehat{S^*}=\widehat{S}\cap\widehat{S'}$.  Hence,
  $(p'_1)_y>\yVal{S}{(p'_1)_x}$, i.e., Statement~(1) of \cref{cor:ibp-} is
  satisfied; and since $p'_1\in U$, Statement~(2) of \cref{cor:ibp-} is
  satisfied.  This, together with \cref{cor:ibp-}, implies $p'_1\in
  \Int(P)$ if and only if $\pari{S}{P}=0$. By analogous arguments, we
  conclude that $(p'_2)_y<\yVal{S}{(p'_2)_x}$, i.e., Statement~(1) of
  \cref{prop:ibp} is satisfied; and since $p'_2\in U$, Statement~(2) of
  \cref{prop:ibp} is satisfied.  Hence, \cref{prop:ibp} implies $p'_2\in
  \Int(P)$ if and only if $\pari{S}{P}=1$.  Since $p'_1,p'_2\notin P$, we
  have $p'_1,p'_2\in \Int(P)\cup \Ext(P)$.  Thus, $p'_2\in\Int(P)$ and
  $p'_1\in\Ext(P)$ if $\pari{S}{P}=1$, and $p'_1\in\Int(P)$ and
  $p'_2\in\Ext(P)$, otherwise (i.e., $\pari{S}{P}=0$).
  
  Independent of $\pari{S}{P}$, therefore there are $p,p'\in
  \widehat{S'}\cap U\subseteq S'\subseteq P'$ with $p\in \Int(P)$ and
  $p'\in\Ext(P)$.  This, \cref{lem:s-o-r} and the fact that $P$ and $P'$
  are overlap-free implies a contradiction.  Hence, we have
  $\yVal{S}{\xi}\le\yVal{S'}{\xi}$ for all $\xi\in \Ih{S}\cap \Ih{S'}$, or
  $\yVal{S'}{\xi}\le\yVal{S}{\xi}$ for all $\xi\in \Ih{S}\cap \Ih{S'}$.
\end{proof}

\begin{remark}
\emph{From here on, all sets $\P$ of polygons are considered to be
overlap-free.}
\end{remark}  

Since a segment $S$ may be part of two or more polygons, we now use pairs
$(S,P)$ in order identify the polygon $P$ of which $S$ is considered to be
a maximal outstretched segment. We write $\SP{\P}$ for the set of all pairs
$(S,P)$ of maximal outstretched segments $S\in \mc{S}$ and their polygons
$P\in \mathcal{P}$.  Moreover, we set 
\[\SPx{\P}{\xi}\coloneqq \big\{(S,P)\in \SP{\P} \mid \xi\in \Ih{S}\big\}.\] 
A segment $S$ that is contained in two polygons
$P$ and $P'$ might be maximal outstretched in $P$ but not in $P'$. Hence,
$(S,P)\in \SP{\P}$ does not necessarily imply $(S,P')\in \SP{\P}$ even if
$S$ is a segment of $P'$ as well. Thus, in general, $\SP{\P} \neq
\mathcal{S}\times \mathcal{P}$. Furthermore, note that the parity of
$S$ depends on $P$, hence we write $\pari{S}{P}$ form here on.

\begin{figure}[t]
  \begin{minipage}{0.55\textwidth}
    \caption{Sweep-adjacency in a set $\P=\{P_1,P_2,P_3,P_4,P_5\}$ of
      polygons.  We have $\Int(P_5)\subsetneq \Int(P_4)$, and
      $\Int(P_i)\cap\Int(P_j)=\emptyset$ for all $i,j \in \{1,\ldots,5\}$
      with $\{i,j\}\ne\{4,5\}$. The parities of their maximal outstretched
      segments $S_i$ are colored red if there is a polygon $P_j$ with
      $\varpi(S_i,P_j)=1$, and are colored blue if there is a polygon $P_j$
      with $\varpi(S_i,P_j)=0$. Note that the parity of segment depends on
      the segments to which it belongs: The $S_2$ appears in both $P_1$ and
      $P_2$ with distinct parities $\varpi(S_2,P_1)=1\neq
      0=\varpi(S_2,P_2)$. Such (parts of) segments are shown as
      blue-red-dashed lines.  Note, furthermore, that $S_4\subseteq S_3$
      with $(S_3,P_2),(S_4,P_3)\in \SP{\P}$.\newline Sweep-adjacency of
      maximal outstretched segments is determined for pairs of polygons:
      First, $(S_2,P_2)$ is only sweep-adjacent to $(S_2,P_1)$ (witnessed
      by $\xi_1$) but not to any other $(S_i,P_j)\in\SP{\P}$.  In contrast,
      $(S_2,P_1)$ is sweep-adjacent to $(S_2,P_2)$ and $(S_4,P_3)$
      (witnessed by $\xi_1$).  Moreover, $(S_5,P_3)$ and $(S_6,P_4)$
      (resp.,~$(S_9,P_3)$ and $(S_8,P_4)$) are sweep-adjacent (witnessed by
      $\xi_2$).  Hence, $(*,P_i)$ and $(*,P_j)$ with $i\ne j$, which are
      sweep-adjacent, are not necessarily unique.  Finally, note that
      $(S_8,P_5)$ and $(S_{11},P_3)$ are sweep-adjacent witnessed by
      $\xi_3$.  However, their sweep-adjacency is neither witness by
      $\xi_2$ (since \cref{def:adjacent}~(2) is violated) nor by $\xi_4$
      (since \cref{def:adjacent}~(1) is violated).}
    \label{fig:sweep-xample}
  \end{minipage} \hfil \begin{minipage}{0.35\textwidth}
    \includegraphics[width=1\textwidth]{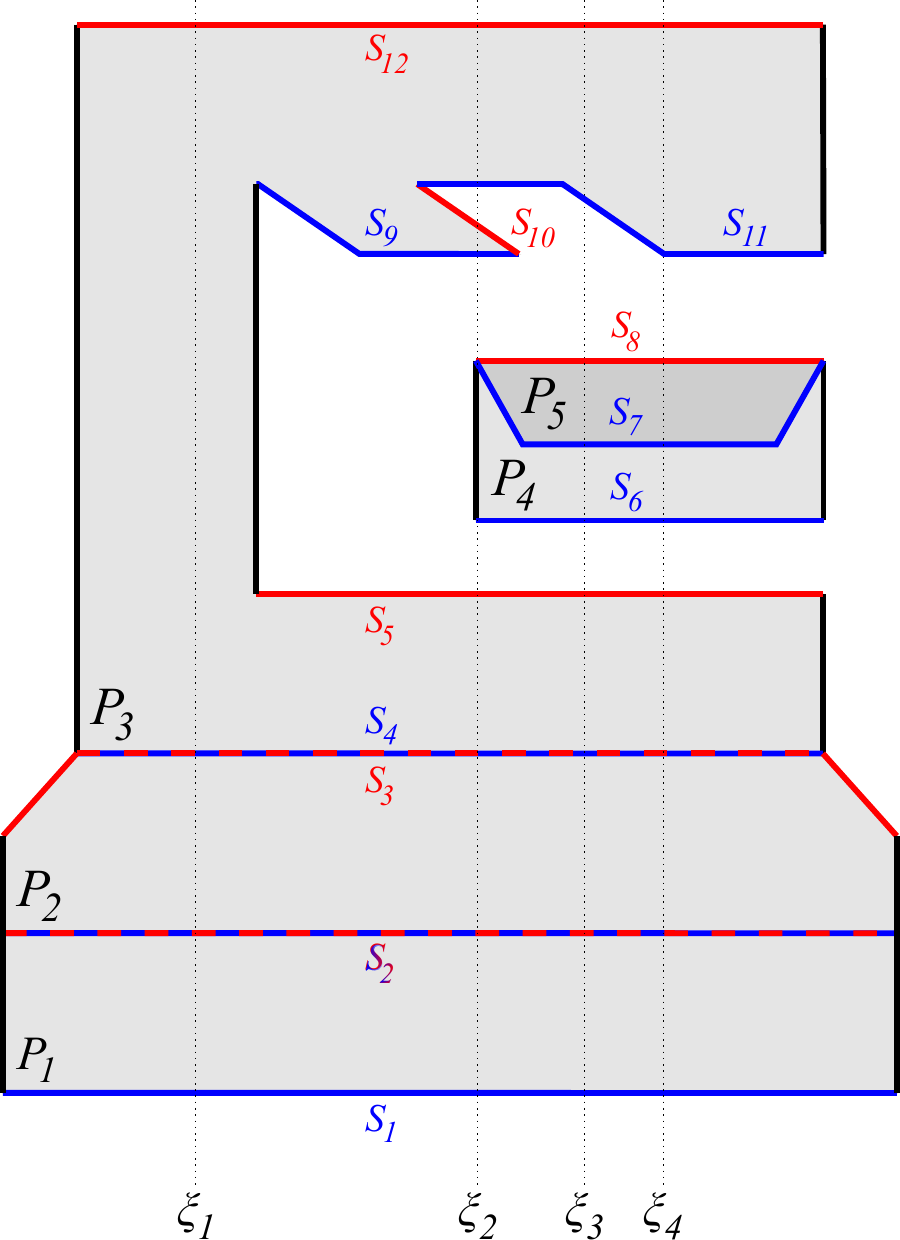}
  \end{minipage}
\end{figure}

In this section, we are interested of the ``relative position'' of two
polygons $P$ and $P'$, i.e. $\Int(P)\subseteq \Int(P')$, $\Int(P)\subseteq
\Int(P')$ or $\Int(P)\cap\Int(P')=\emptyset$.  To this end, we introduce
the concept of ``sweep-adjacency'', which is illustrated in
\cref{fig:sweep-xample}.
\begin{definition}\label{def:adjacent}
  Let $\P$ be a set of polygons and $(S,P),(S',P')\in\SP{\P}$ with $P\ne
  P'$.  Then, we call $(S,P)$ and $(S',P')$ \emph{sweep-adjacent} if there
  is a $\xi\in \Io{S}\cap \Io{S'}$ such that there is no $(S'',P'')\in
  \SPx{\P}{\xi}$ with $P''\in \{P,P'\}$, and
  \begin{equation*}
  \yVal{S}{\xi}<\yVal{S''}{\xi}<\yVal{S'}{\xi}\textnormal{ or }
  \yVal{S}{\xi}>\yVal{S''}{\xi}>\yVal{S'}{\xi}.
  \end{equation*}
  In this case, we say that $\xi$ is a \emph{witness} for the
  sweep-adjacency of $(S,P)$ and $(S',P')$.
\end{definition}
Note that $\min(S)_x$ and $\min(S')_x$ need to be excluded from
$\Io{S}\cap\Io{S'}$, since there is another segment $(S'',P)\in\SP{\P}$
with $\min(S'')=\min(S)$ (resp., $(S'',P')\in\SP{\P}$ with
$\min(S'')=\min(S')$), which would lead to ambiguities.
Moreover, in order to get some intuition, we observe the following:

\begin{obs} \label{obs:alwaysone}
  Let $\P$ be a set of polygons, $P,P'\in\P$ with $P\ne P'$, and
  $(S,P),(S',P')\in \SP{\P}$.  Then, the existence of a $\xi\in \Io{S}\cap
  \Io{S'}$ with $\yVal{S}{\xi}=\yVal{S'}{\xi}$ implies (a) that
  $p\coloneqq(\xi,\yVal{S}{\xi})=(\xi,\yVal{S'}{\xi})\in S\cap S'$, and (b)
  that $(S,P)$ and $(S',P')$ are sweep-adjacent.  However, $p\in S\cap S'$
  with $p_x\in\Io{S}\cap\Io{S'}$ does \emph{not} imply that $(S,P)$ and
  $(S',P')$ are sweep-adjacent; see \cref{fig:non-sweep} for an
  illustrative example.
\end{obs}

\begin{lemma}\label{lem:sweep-adja-pair}
  Let $\P$ be a set of polygons and $P,P'\in\P$. Then, there is a
  sweep-adjacent pair $(S,P),(S',P')\in \SP{\P}$ if and only if $\Ih{P}\cap
  \Ih{P'}\ne\emptyset$ and $P\neq P'$.
\end{lemma}
\begin{proof}
  Let $\P$ be a set of polygons, $P,P'\in\P$ with $P\ne P'$.  If
  $\Ih{P}\cap \Ih{P'}=\emptyset$, then there is no $\xi$ such that
  $(S,P),(S'P')\in \SPx{\P}{\xi}$, and thus, in particular, no
  sweep-adjacent pair of segments exists. Conversely, we assume $\Ih{P}\cap
  \Ih{P'}\ne\emptyset$. Since $\Ih{P}$ and $\Ih{P'}$ are semi-open to the
  right, so is their intersection, and thus there is an open interval
  $J\subseteq \Ih{P}\cap \Ih{P'}\ne\emptyset$. Moreover, $J'\coloneqq
  J\setminus\{v_x\mid v\in V(P)\cup V(P')\}\ne\emptyset$, and for every
  $\xi\in J'$, we have $\mc{S}\coloneqq\{S\mid (S,P)\in \SPx{\P}{\xi}\}$
  and $\mc{S'}\coloneqq\{S'\mid (S',P')\in \SPx{\P}{\xi}\}$ are non-empty.
  Now, choose $S\in \mc{S}$ and $S'\in \mc{S'}$ such that
  $|\yVal{S}{\xi}-\yVal{S'}{\xi}|$ is minimal.  It follows immediately that
  $(S,P)$ and $(S',P')$ are sweep-adjacent.
\end{proof}
\begin{lemma}\label{lem:sw+delta}
  Let $\P$ be a set of polygons, and let $(S,P),(S',P')\in \SP{\P}$ be
  sweep-adjacent witnessed by $\xi$.  Then, there exists a sufficiently
  small $\delta>0$ such that sweep-adjacency is also witnessed by 
  $\xi+\delta$ and $\xi-\delta$.
  In particular, there is always a witness $\xi'\in
  \Io{S}\cap\Io{S'}$ with $\xi \ne v_x$ for all $v \in V(P)\cup V(P')$.
\end{lemma}
\begin{proof}
  Let $\P$ be a set of polygons, and let $(S,P),(S',P')\in \SP{\P}$ be a
  sweep-adjacent pair witnessed by $\xi$ and we assume
  w.l.o.g.\ $\yVal{S}{\xi}\le\yVal{S'}{\xi}$.  Hence, \cref{def:adjacent}
  implies $\xi\in \Io{S}\cap\Io{S'}$. This, together with
  $\Io{S}\cap\Io{S'}$ being an open interval, implies that there is a
  sufficiently small open neighborhood $U$ of $\xi$ such that $U\subseteq
  \Io{S}\cap\Io{S'}$ and $\SPx{\P}{\xi}=\SPx{\P}{\xi'}$ for all $\xi'\in
  U$.  In particular, there is a $\delta>0$ such that (a) $\xi\pm\delta\in
  U$ and b $\xi\pm\delta\ne v_x$ for all $v\in V(P)\cup V(P')$.  Moreover,
  by \cref{def:adjacent}, there is no $(S'',P'')\in\SP{\P}$ with $P''\in
  \{P,P'\}$ and $\yVal{S}{\xi}<\yVal{S''}{\xi}<\yVal{S'}{\xi}$.  This,
  together with $\SPx{\P}{\xi}=\SPx{\P}{\xi\pm\delta}$ and
  \cref{lem:mos-le-1}, implies that there is no $(S'',P'')\in\SP{\P}$ with
  $P''\in \{P,P'\}$ and $\yVal{S}{\xi'}<\yVal{S''}{\xi'}<\yVal{S'}{\xi'}$.
  Hence, $\xi\pm\delta$ is also a witness for the sweep-adjacency of
  $(S,P)$ and $(S',P')$.
\end{proof}

\begin{figure}[t]
  \begin{minipage}{0.65\textwidth}
    \caption{Intersection of maximal outstretched segments does not imply
      sweep-adjacency as shown by the maximal outstretched segments $S_1$
      of $P_1$ and $S_2$ of $P_2$. We observe $\Int(P_1)\cap
      \Int(P_2)=\emptyset$ and $\Ih{P_1}=\Ih{P_2}$, and $V(S_1)\cap
      V(S_2)\ne \emptyset$ and $E(S_1)\cap E(S_2)\ne\emptyset$.
      \newline It is easy to verify that for every $\xi \in
      \Ih{S_1}\cap\Ih{S_2}$, we have either
      $\yVal{S_1}{\xi}>\yVal{S'_1}{\xi}>\yVal{S_2}{\xi}$ (if $p_x\le \xi$),
      or $\yVal{S_1}{\xi}>\yVal{S'_2}{\xi}>\yVal{S_2}{\xi}$ (otherwise).
      Hence, for all $\xi\in \Ih{S_1}\cap\Ih{S_2}$, the condition of
      \cref{def:adjacent} is not satisfied; and thus, $(S_1,P_1)$ and
      $(S_2,P_2)$ cannot be sweep-adjacent.  Hence, neither $V(S_1)\cap
      V(S_2)\ne \emptyset$ nor $E(S_1)\cap E(S_2)\ne\emptyset$ is not
      sufficient to imply sweep-adjacency.}
    \label{fig:non-sweep}
  \end{minipage} \hfil \begin{minipage}{0.25\textwidth}
    \includegraphics[width=1\textwidth]{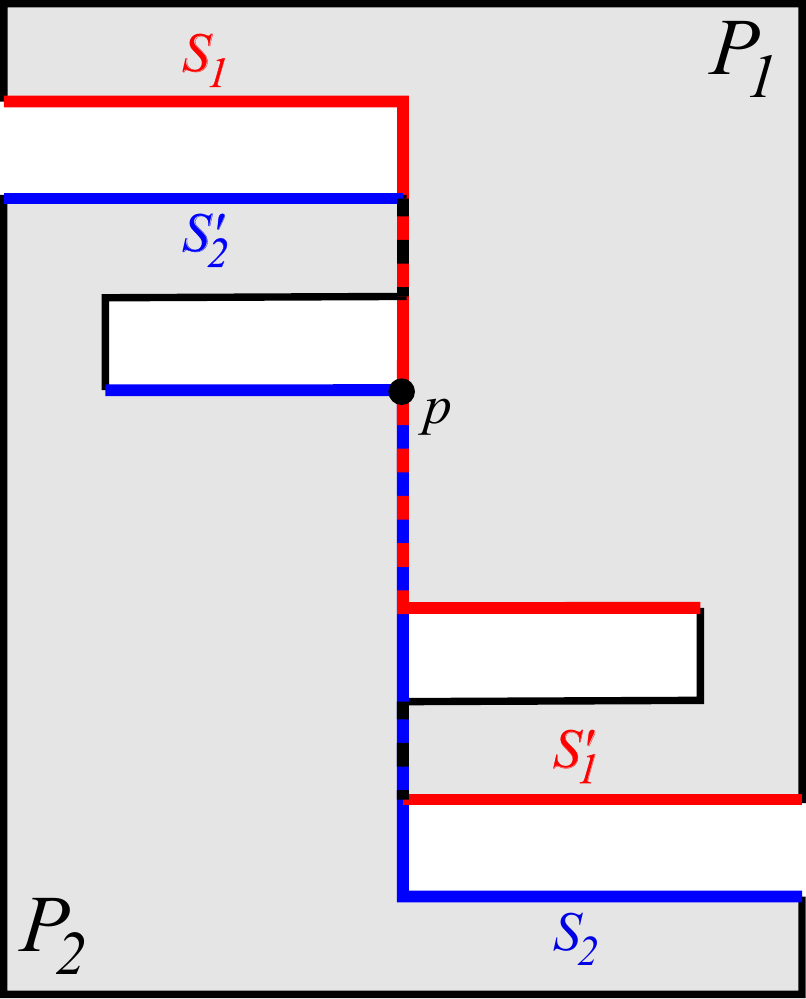}
  \end{minipage}
\end{figure}

Now, we proceed by showing how sweep-adjacent maximally outstretched segments
can be used to determine the relative location of their polygons. First, we
characterize when two polygons have a (non-)empty intersection of their 
interiors.
\begin{lemma} \label{lem:adj-overlap}
  Let $\P$ be a set of polygons, and $(S,P),(S',P')\in\SP{\P}$ 
  be sweep-adjacent. Then, $\Int(P)\cap\Int(P')\ne\emptyset$ 
  if and only if $\varpi(S,P)=\varpi(S',P')$.
\end{lemma}
\begin{proof}
  Let $\P$ be a set of polygons, and let $(S,P),(S',P')\in \SP{\P}$ be
  sweep-adjacent.  Then, by \cref{lem:sw+delta}, the sweep-adjacency can be
  witnessed by a $\xi\in \Io{S}\cap\Io{S'}$ with $\xi\ne v_x$ for all $v\in
  V(S)\cup V(S')$. Moreover, we may assume w.l.o.g.\ that $\yVal{S}{\xi}\le
  \yVal{S'}{\xi}$.  In the following, we consider the points $p\coloneqq
  (\xi,\yVal{S}{\xi}-\delta)$ and $p'\coloneqq (\xi,\yVal{S'}{\xi}+\delta)$
  for some sufficiently small $\delta>0$.  By the choice of $\xi$, we have
  $\xi\ne v_x$ for all $v\in V(S)\cup V(S')$. This implies that $p$ and
  $p'$ cannot be on a vertical edge of $P$ or $P'$. Hence, since $\delta>0$
  is sufficiently small, we have $p,p'\notin P\cup P'$.  Since $(S,P)$ and
  $(S',P')$ are sweep-adjacent, it is easy to verify that $p$ and $p'$
  always satisfies Conditions~(1) and~(2) of \Cref{prop:ibp} and
  \Cref{cor:ibp-}, respectively.  Hence, $p\in\Int(P)$ iff $\varpi(S,P)=1$
  and $p\in\Int(P')$ iff $\varpi(S',P')=1$, resp., $p'\in\Int(P)$ iff
  $\varpi(S,P)=0$ and $p'\in\Int(P')$ iff $\varpi(S',P')=0$.  Thus, in
  either case, $\varpi(S,P)=\varpi(S',P')$ implies
  $\Int(P)\cap\Int(P')\ne\emptyset$.

  Now, suppose $\varpi(S,P)\ne\varpi(S',P')$.  Then, we distinguish two
  cases: (i) $\varpi(S,P)=0$ and $\varpi(S',P')=1$, and (ii)
  $\varpi(S,P)=1$ and $\varpi(S',P')=0$.  In Case (i), we have
  $p'\in\Int(P)$, $p\notin\Int(P)$, $p\in\Int(P')$, and $p\notin\Int(P')$,
  and thus $p\in \Int(P')\setminus\Int(P)$ and $p'\in
  \Int(P)\setminus\Int(P')$.  In Case (ii), we have $p\in\Int(P)$ and
  $p'\notin\Int(P)$, $p'\in\Int(P')$ and $p\notin\Int(P')$, i.e.,
  $p\in\Int(P)\setminus\Int(P')$ and $p'\in\Int(P')\setminus\Int(P)$. In
  both cases, the fact that $\Int(P)$ and $\Int(P')$ are overlap-free
  implies that $\Int(P)\cap\Int(P')=\emptyset$.  Therefore,
  $\varpi(S,P)\ne\varpi(S',P')$ implies $\Int(P)\cap\Int(P')=\emptyset$.
\end{proof}
The combination of \cref{obs:alwaysone} and \Cref{lem:adj-overlap} shows
that the intersection of two polygons is determined by the parity of a pair
sweep-adjacent maximally outstretched segments or their non-existence.
\Cref{lem:adj-overlap,lem:area} together imply
\begin{corollary} \label{cor:area}
  Let $\P$ be a set of polygons, and suppose that $(S,P),(S',P')\in\SP{\P}$
  are sweep-adjacent.  Then, $\Int(P)\subsetneq\Int(P')$ if and only if
  $\varpi(S,P)=\varpi(S',P')$ and $\area{P}<\area{P'}$.
\end{corollary}

We formalize the idea that segments are ``below'' of other segments in the
following manner:
\begin{definition}
  Let $\P$ be a set of polygons, $(S,P), (S',P')\in\SP{\P}$ and $\Ih{S}\cap
  \Ih{S'}\ne\emptyset$. Note that $P=P'$ is possible.  Then, $(S,P)$ is
  \emph{below} $(S',P')$ if there is $\xi\in \Ih{S}\cap \Ih{S'}$ such that
  $\yVal{S}{\xi}<\yVal{S'}{\xi}$.
\end{definition}

\begin{obs}\label{obs:below}
Note that if $\yVal{S}{\xi}\ne\yVal{S'}{\xi}$ for some $\xi\in \Ih{S}\cap
\Ih{S'}$, then we have by \cref{lem:mos-le-1} \emph{either} $(S,P)$ is
below $(S',P')$, or \emph{vice versa}. In this case, we have $S\neq S'$.
\end{obs}

\begin{proposition} \label{lem:below-parity}
  Let $\P$ be a set of polygons, let $(S,P), (S',P')\in\SP{\P}$ be
  sweep-adjacent as witnessed by some $\xi$ with $\xi\ne v_x$ for all $v\in
  V(S)\cup V(S')$, and $\yVal{S}{\xi}\ne\yVal{S'}{\xi}$.  Then, we have
  $\yVal{S}{\xi}<\yVal{S'}{\xi}$ if and only if one of the following
  statements hold:
  \begin{enumerate}[nolistsep,noitemsep]
  \item $\varpi(S,P)=1$ and $\varpi(S',P')=0$;
  \item $\varpi(S,P)=\varpi(S',P')=1$ and $\area{P}<\area{P'}$, or
  \item $\varpi(S,P)=\varpi(S',P')=0$ and $\area{P}>\area{P'}$.
\end{enumerate}
\end{proposition}
\begin{proof}
  Let $\P$ be a set of polygons, let $(S,P), (S',P')\in\SP{\P}$ be
  sweep-adjacent as witnessed by some $\xi$, and
  $\yVal{S}{\xi}<\yVal{S'}{\xi}$.  Then, consider the points
  $p\coloneqq(\xi,\yVal{S}{\xi}-\delta)$,
  $p'\coloneqq(\xi,\yVal{S'}{\xi}+\delta)$ for sufficiently small
  $\delta>0$, and
  $p''\coloneqq(\xi,\frac{1}{2}(\yVal{S}{\xi}+\yVal{S'}{\xi}))$. By
  assumption, we have $\xi\ne v_x$ for all $v\in V(S)\cup V(S')$.  This
  implies that $p$ and $p'$ cannot be on a vertical edge of $P$ or $P'$.
  Hence, since $\delta>0$ is sufficiently small, we have $p,p'\notin P\cup
  P'$.  Since $(S,P)$ and $(S',P')$ are sweep-adjacent witnessed by $\xi$,
  we can immediately apply \Cref{prop:ibp} and \Cref{cor:ibp-} to determine
  whether or not $p$, $p'$, and $p''$ are located within $\Int(P)$ and
  $\Int(P')$.  Then, assume for contradiction that $\varpi(S,P)=0$ and
  $\varpi(S',P')=1$.  In this case, it holds that $p'\in\Int(P)$,
  $p\notin\Int(P)$ and $p''\in\Int(P)$, as well as $p\in\Int(P')$,
  $p'\notin\Int(P')$ and $p''\in\Int(P')$.  Thus, $p''\in
  \Int(P)\cap\Int(P')$, $p\in \Int(P')\setminus\Int(P)$ and $p'\in
  \Int(P)\setminus\Int(P')$ implies a contradiction, since $P$ and $P'$ are
  overlap-free.  Hence, if $\pari{S}{P}\ne \pari{S'}{P'}$, then
  $\yVal{S}{\xi}<\yVal{S'}{\xi}$ iff $\varpi(S,P)=1$ and $\varpi(S',P')=0$.
  Moreover, if $\varpi(S,P)=\varpi(S',P')=1$, then $p\in\Int(P)$,
  $p,p''\in\Int(P')$ and $p''\notin\Int(P)$, which is true iff
  $\Int(P)\subsetneq\Int(P')$, and thus iff
  $\Int(P)\cap\Int(P')\ne\emptyset$ and $\area{P}<\area{P'}$.  Hence, if
  $\varpi(S,P)=\varpi(S',P')=1$, then $\yVal{S}{\xi}<\yVal{S'}{\xi}$ iff
  $\area{P}<\area{P'}$.
  Finally, if $\varpi(S,P)=\varpi(S',P')=0$, then $p'\in\Int(P')$,
  $p',p''\in\Int(P)$, $p''\notin\Int(P')$, which is true iff
  $\Int(P')\subsetneq\Int(P)$, and thus iff
  $\Int(P)\cap\Int(P')\ne\emptyset$ and $\area{P'}<\area{P}$.  Hence, if
  $\varpi(S,P)=\varpi(S',P')=0$, then $\yVal{S}{\xi}<\yVal{S'}{\xi}$ iff
  $\area{P}>\area{P'}$.
\end{proof}

\begin{corollary} \label{cor:below-parity2}
  Let $\P$ be a set of polygons, and let $(S,P),(S',P')\in \SP{\P}$ be
  sweep-adjacent.  If $(S,P)$ is below $(S',P')$, then $(S,P)$ and
  $(S',P')$ satisfy one of the Statements~(1) to~(3) in
  \Cref{lem:below-parity}.
\end{corollary}
\begin{proof}
  Let $\P$ be a set of polygons, and let $(S,P),(S',P')\in \SP{\P}$ be
  sweep-adjacent.  Then, by \cref{lem:sw+delta}, the sweep-adjacency can be
  witnessed by a $\xi\in \Io{S}\cap\Io{S'}$ with $\xi\ne v_x$ for all $v\in
  V(S)\cup V(S')$. Moreover, by definition of sweep-adjacency, $P\ne P'$.
  We assume that $(S,P)$ is below $(S',P')$.  If $\yVal{S}{\xi}\ne
  \yVal{S'}{\xi}$, then we can use \Cref{lem:below-parity}, and we are
  done.  Now, assume that $\yVal{S}{\xi}= \yVal{S'}{\xi}$.  Since $(S,P)$
  is below $(S',P')$, there is a $\xi'\in \Ih{S}\cap \Ih{S'}$ with
  $\yVal{S}{\xi'}\ne \yVal{S'}{\xi'}$.  We may assume w.l.o.g.\ that
  $\xi<\xi'$.  Then, consider $\xi''\in [\xi,\xi']$ such that $|\xi''-\xi|$
  is maximal and $\yVal{S}{\xi'''}=\yVal{S'}{\xi'''}$ for all
  $\xi'''\in[\xi,\xi'']$.  In particular, for a sufficiently small
  $\delta>0$, the sweep-adjacency of $(S,P)$ and $(S',P')$ can be witnessed
  by $\xi''+\delta$.  By the choice of $\xi''+\delta$, we conclude that
  $\yVal{S}{\xi''+\delta}\ne \yVal{S'}{\xi''+\delta}$, and $\xi''+\delta\ne
  v_x$ for all $v\in V(S)\cup V(S')$.  Hence, \Cref{lem:below-parity} can
  be applied.
\end{proof}

Recall that the intervals $\Ih{e}$ of the non-vertical edges $e$ of an 
outstretched segment $S$ of $P$ form a partition of $\Ih{S}$. 
Hence, for every $\xi\in \Ih{S}$,
there is a unique non-vertical edge $e\in E(S)$ such that $\xi\in \Ih{e}$.
Therefore, we define the \emph{slope} of $S$ for each $\xi\in \Ih{S}$ as the
slope of the unique non-vertical edge $e=\{u,v\}$ with $\xi\in \Ih{e}$ as
\begin{equation*}
  \slope{S}{\xi} \coloneqq \frac{u_y-v_y}{u_x-v_x}.
\end{equation*}
The slope $\slope{S}{\xi}$ can be used to determine, for two outstretched
segments $(S,P)$ and $(S',P')$ 
that share a common point $(x,y)$ with $x\in \Ih{S}\cap\Ih{S'}$
whether $S$ lies above or below $S'$. To be more precise:
\begin{lemma} \label{lem:slope}
  Let $\P$ be a set of polygons, and let $(S,P),(S',P')\in \SP{\P}$ such that 
  $\yVal{S}{\xi}=\yVal{S'}{\xi}$ and $\slope{S}{\xi}\ne\slope{S'}{\xi}$ for
  some $\xi\in \Ih{S}\cap\Ih{S'}$.
  Then, $\slope{S}{\xi}<\slope{S'}{\xi}$ if and only if $(S,P)$ is below
  $(S',P')$.
\end{lemma}
\begin{proof}
  Let $(S,P),(S',P')\in \SP{\P}$ as specified in the lemma.
  First, assume that $\slope{S}{\xi}<\slope{S'}{\xi}$. 
  In particular,  for sufficiently 
  small $\delta>0$, we have $\xi'\coloneqq \xi+\delta\in \Ih{S}\cap \Ih{S'}$.
  Therefore, we have
  $\yVal{S}{\xi'}=y+\slope{S}{\xi}\cdot\delta <
  y+\slope{S'}{\xi}\cdot\delta = \yVal{S'}{\xi'}$, i.e., $(S,P)$ is below
  $(S',P')$. 
  For the converse, assume by contraposition that 
  $\slope{S}{\xi}>\slope{S'}{\xi}$.
  By similar arguments, $(S',P')$ must be below $(S,P)$. 
  By \Cref{obs:below}, $(S,P)$ cannot be below $(S',P')$.
\end{proof}

\section{Ordering the Set of Maximally Outstretched Segments}
\label{sec:order}

The discussion in \cref{sec:invest} motivates the definition of the
following order of maximally outstretches segments along a sweep-line at
$\xi$:
\begin{definition}\label{def:ONFP}
  Let $\P$ be a set of polygons and $\xi\in\mathbb{R}$. Then, we define the
  relation $\lessdot_\xi$ on $\SPx{\P}{\xi}$ by setting $(S',P')
  \lessdot_{\xi} (S,P)$ whenever 
    \begin{itemize}[noitemsep,nolistsep]
  \item[\textbf{(1a)}] $\yVal{S'}{\xi}>\yVal{S}{\xi}$, or
  \item[\textbf{(1b)}] $\yVal{S'}{\xi}=\yVal{S}{\xi}$ and 
    \begin{itemize}[noitemsep,nolistsep]
    \item[\textbf{(2a)}] $\slope{S'}{\xi}>\slope{S}{\xi}$, or
    \item[\textbf{(2b)}] $\slope{S'}{\xi}=\slope{S}{\xi}$  and 
      \begin{itemize}[noitemsep,nolistsep]
      \item[\textbf{(3a)}] $\pari{S'}{P'}=0$ and $\pari{S}{P}=1$, or
      \item[\textbf{(3b)}] $\pari{S'}{P'}=\pari{S}{P}=1$ and
        $\area{P}<\area{P'}$, or  
      \item[\textbf{(3c)}] $\pari{S'}{P'}=\pari{S}{P}=0$ and
        $\area{P}>\area{P'}$.
      \end{itemize}
    \end{itemize}
\end{itemize}
\end{definition}
For an illustrative example of \cref{def:ONFP}, see
\cref{fig:sweep-xample}. There, we have
$(S_{12},P_3)\lessdot_{\xi_1}(S_{1},P_1)$ due to \cref{def:ONFP}~(1),
$(S_{8},P_5)\lessdot_{\xi_2}(S_{7},P_5)$ due to \cref{def:ONFP}~(2a),
$(S_4,P_3)\lessdot_{\xi_2}(S_3,P_2)$ due to \cref{def:ONFP}~(3a), and
$(S_8,P_4)\lessdot_{\xi_3}(S_8,P_5)$ due to \cref{def:ONFP}~(3b).  Next, we
proceed by showing that $\lessdot_\xi$ is a total order on the set of
maximal outstretched segments intersecting a sweep-line at $\xi$. To this
end, we first express equality in terms of the quantities appearing in
\cref{def:ONFP}.
\begin{lemma}\label{obs:lessdot-equal}
  Let $\P$ be a set of polygons and
  $(S,P),(S',P')\in \SPx{\P}{\xi}$.  Then, $(S,P)=(S',P')$ if and only if
  $\yVal{S}{\xi}=\yVal{S'}{\xi}$, $\slope{S}{\xi}=\slope{S'}{\xi}$,
  $\pari{S}{P}=\pari{S'}{P'}$ and $\area{P}=\area{P'}$.
\end{lemma}
\begin{proof}
  Clearly, the condition is necessary.  To establish the sufficiency, we
  consider the shared point $p\coloneqq
  (\xi,\yVal{S}{\xi})=(\xi,\yVal{S'}{\xi})$ and the same slope
  $\slope{S}{\xi}=\slope{S'}{\xi}$.  For sufficiently small $\delta>0$, we
  have $p'\coloneqq(\xi',\yVal{S}{\xi'})= (\xi',\yVal{S'}{\xi'})$ with
  $\xi'\coloneqq\xi+\delta$.  In particular, 
  since $(S,P),(S',P')\in \SPx{\P}{\xi}$ and $\delta$ is sufficiently
  small, we have
  $\xi'\in \Ih{S}\cap \Ih{S'}$ and
  $\min(S)_x\le \xi<\xi+\delta=\xi'<\max(S)_x$ and $\min(S')_x\le
  \xi<\xi+\delta=\xi'<\max(S')_x$.  Hence, \cref{obs:alwaysone} implies
  that $(S,P)$ and $(S',P')$ are sweep-adjacent. This, together with
  $\pari{S}{P}=\pari{S'}{P'}$ and \Cref{lem:adj-overlap}, implies
  $\Int(P)\cap\Int(P')\ne\emptyset$.  Now, $\area{P}=\area{P'}$ and
  \Cref{lem:area} imply $P=P'$.  Finally, if $S$ and $S'$ are two maximal
  outstretched segments of the same polygon $P=P'$ with the common vertex
  $p'\in S\cap S'$, then \Cref{lem:bSPp} implies $S=S'$. Hence,
  $(S,P)=(S',P')$.
\end{proof}

\begin{lemma}
  Let $\P$ be a set of polygon. Then, the relation
  $\lessdot_{\xi}$ specified in \cref{def:ONFP} is a strict total order on
  $\SPx{\P}{\xi}$ for all  $\xi\in\mathbb{R}$.
\end{lemma}
\begin{proof}
  Let $\P$ be a set of polygon and $\xi\in\mathbb{R}$, and let
  $(S,P),(S',P'),(S'',P'')\in \SPx{\P}{\xi}$.  First, it is obvious from
  the definition that $\lessdot_{\xi}$ is antisymmetric, that is,
  $(S,P)\lessdot_{\xi}(S',P')$ implies that $(S',P')\lessdot_{\xi}(S,P)$
  must not hold.  Moreover, we have neither $(S,P)\lessdot_{\xi}(S',P')$
  nor $(S',P')\lessdot_{\xi}(S,P)$ if and only
  $\yVal{S}{\xi}=\yVal{S'}{\xi}$, $\slope{S}{\xi}=\slope{S'}{\xi}$,
  $\pari{S}{P}=\pari{S'}{P'}$ and $\area{P}=\area{P'}$. By
  \Cref{obs:lessdot-equal}, this is exactly the case if $(S,P)=(S',P')$.
  Thus, $\lessdot_{\xi}$ is trichotomous, that is, for all
  $(S,P),(S',P')\in \SPx{\P}{\xi}$ we have \emph{either} \textit{(i)}
  $(S',P')\lessdot_{\xi}(S,P)$, \textit{(ii)} $(S,P)\lessdot_{\xi}(S',P')$
  or \textit{(iii)} $(S',P')=(S,P)$.
 
  It remains to show that $\lessdot_{\xi}$ is transitive. To this end, let
  $(S,P)\lessdot_{\xi} (S',P')\lessdot_{\xi} (S'',P'')$. If the first or
  the second inequality is realized by Case~(1a) or~(1b) and~(2a), then
  $(S,P)\lessdot_{\xi}(S'',P'')$ follows from transitivity of $<$ on
  $\mathbb{R}$. It remains to consider
  $\yVal{S}{\xi}=\yVal{S'}{\xi}=\yVal{S''}{\xi}$ and
  $\slope{S}{\xi}=\slope{S'}{\xi}=\slope{S''}{\xi}$.  First, assume
  $\varpi(S,P)=0$ and $\varpi(S',P')=1$. Then,
  $(S',P')\lessdot_{\xi}(S'',P'')$ leaves only case (3b), i.e.,
  $\varpi(S',P')=(S'',P'')=1$ and thus, $(S,P)\lessdot_{\xi}(S'',P'')$
  according to (3a). If $\varpi(S',P')=0$ and $\varpi(S'',P'')=1$, then
  $(S,P)\lessdot_{\xi}(S',P')$ implies, via (3c)
  $\varpi(S,P)=\varpi(S',P')=0$ and $(S,P)\lessdot_{\xi}(S'',P'')$ by
  (3a). Note that, by definition of $\lessdot_{\xi}$, the case
  $\varpi(S,P)=1$ and $\varpi(S,P)\neq \varpi(S',P')$ or $\varpi(S,P)\neq
  \varpi(S'',P'')$ cannot occur.  Hence, it remains to consider the case
  that $\varpi(S,P)=\varpi(S',P')=\varpi(S'',P'')$. For even parity
  $\varpi(S,P)=0$, we have $\area{P}<\area{P'}<\area{P''}$ and thus
  $(S,P)\lessdot_{\xi}(S'',P'')$ by (3c).  For odd parity $\varpi(S,P)=1$,
  we have $\area{P}>\area{P'}>\area{P''}$ and thus
  $(S,P)\lessdot_{\xi}(S'',P'')$ by (3b).
\end{proof}

Next, we show that $\lessdot_{\xi}$ is consistent with the notion of a
maximally outstretched segment being below another.
\begin{lemma} \label{lem:bel-iff-lessdot}
  Let $\P$ be a set of polygons and $(S,P),(S',P')\in \SP{\P}$.  If there
  is some $\xi\in \Ih{S}\cap \Ih{S'}$ with
  $\yVal{S}{\xi}\ne\yVal{S'}{\xi}$, then $(S,P)$ is below $(S',P')$ if and
  only if $(S',P')\lessdot_\xi(S,P)$ for all $\xi\in \Ih{S}\cap \Ih{S'}$.
  In particular, for every $\xi,\xi'\in \Ih{S}\cap \Ih{S'}$, we have
  $(S',P')\lessdot_\xi(S,P)$ if and only if $(S',P')\lessdot_{\xi'}(S,P)$.
\end{lemma}
\begin{proof}
  Let $(S,P),(S',P')\in \SP{\P}$ be specified as in the lemma.

  First, we assume that $(S,P)$ is below $(S',P')$ and thus,
  $(S,P)\neq(S',P')$.  Let $\xi\in \Ih{S}\cap \Ih{S'}$ be chosen
  arbitrarily.  If $\yVal{S}{\xi}\ne\yVal{S'}{\xi}$, then
  $\yVal{S}{\xi}<\yVal{S'}{\xi}$, and thus, $(S',P')\lessdot_\xi(S,P)$.
  Moreover, if $\yVal{S'}{\xi}=\yVal{S}{\xi}$ and
  $\slope{S'}{\xi}\ne\slope{S}{\xi}$, then \Cref{lem:slope} implies
  $\slope{S}{\xi}<\slope{S'}{\xi}$, and thus, $(S',P')\lessdot_\xi(S,P)$.
  Now, assume that $\yVal{S'}{\xi}=\yVal{S}{\xi}$ and
  $\slope{S'}{\xi}=\slope{S}{\xi}$.  Hence, there is a $p\in S\cap S'$ with
  $p_x\ne v_x$ for all $v\in V(S)\cup V(S')$.  Thus, if $P=P'$, then we
  would conclude by \cref{lem:bSPp} that $S=S'$; a contradiction to
  $(S',P')\ne (S,P)$.  Hence, $P\ne P'$.  The latter two observations allow
  us to use \cref{obs:alwaysone}, and we conclude that $(S,P)$ and
  $(S',P')$ are sweep-adjacent.  Hence, by virtue of
  \Cref{cor:below-parity2}, $(S,P)$ below $(S',P')$ implies
  $\pari{S'}{P'}=0$ and $\pari{S}{P}=1$, or $\pari{S'}{P'}=\pari{S}{P}=1$
  and $\area{P}<\area{P'}$, or $\pari{S'}{P'}=\pari{S}{P}=0$ and
  $\area{P}>\area{P'}$, and thus $(S',P')\lessdot_\xi(S,P)$.  In either
  case, we have $(S',P')\lessdot_\xi(S,P)$.  Conversely, assume that
  $(S',P')\lessdot_\xi(S,P)$ for all $\xi\in \Ih{S}\cap \Ih{S'}$.  In
  particular, we have $(S',P')\lessdot_\xi(S,P)$ for $\xi\in \Ih{S}\cap
  \Ih{S'}$ with $\yVal{S}{\xi}\ne\yVal{S'}{\xi}$.  Hence,
  $\yVal{S}{\xi}<\yVal{S'}{\xi}$, and thus, $(S,P)$ is below $(S',P')$.

  Now, let $\xi,\xi'\in \Ih{S}\cap \Ih{S'}$. If there is a $\xi''\in
  \Ih{S}\cap \Ih{S'}$ with $\yVal{S}{\xi''}\ne\yVal{S'}{\xi''}$, then
  either \textit{(i)} $\yVal{S'}{\xi''}<\yVal{S}{\xi''}$ or \textit{(ii)}
  $\yVal{S'}{\xi''}>\yVal{S}{\xi''}$.  In Case~(i) $(S,P)$ is below
  $(S',P')$, and in Case~(ii) $(S',P')$ is below $(S,P)$. Hence,
  application of the afore-established results show that, in Case (i), we
  have $(S',P')\lessdot_\xi(S,P)$ and $(S',P')\lessdot_{\xi'}(S,P)$ and, in
  Case (ii), $(S,P)\lessdot_\xi(S',P')$ and $(S,P)\lessdot_{\xi'}(S',P')$.
  Now, assume that $\yVal{S}{\xi''}=\yVal{S'}{\xi''}$ for every $\xi''\in
  \Ih{S}\cap \Ih{S'}$.  Then, it is easy to see that
  $\yVal{S}{\xi}=\yVal{S'}{\xi}$, $\slope{S}{\xi}=\slope{S'}{\xi}$,
  $\yVal{S}{\xi'}=\yVal{S'}{\xi'}$ and $\slope{S}{\xi'}=\slope{S'}{\xi'}$.
  This, together with the definition of $\lessdot_\xi$ and the fact that
  the parity $\varpi(S,P)$ and $\varpi(S',P')$ is independent from $\xi$
  and $\xi'$, implies $(S',P')\lessdot_\xi(S,P)$ if and only if
  $(S',P')\lessdot_{\xi'}(S,P)$.  In summary, for all $\xi,\xi'\in
  \Ih{S}\cap\Ih{S'}$, we have $(S',P')\lessdot_\xi(S,P)$ if and only if
  $(S',P')\lessdot_{\xi'}(S,P)$.
\end{proof}

Note that \Cref{lem:bel-iff-lessdot} requires that there is a $\xi\in
\Ih{S}\cap \Ih{S'}$ with $\yVal{S}{\xi}\ne\yVal{S'}{\xi}$. However,
segments may be $\lessdot_{\xi}$ comparable without one being below the
other. This is (precisely) the case if $S$ and $S'$ coincide on $\Ih{S}\cap
\Ih{S'}$.  In particular, the total order $\lessdot_{\xi}$ on
$\SPx{\P}{\xi}$ extends to a partial order $\lessdot$ on $\SP{\P}$ by
setting $(S,P)\lessdot(S',P')$ whenever $(S,P)\lessdot_{\xi}(S',P)$ for
some $\xi\in \Ih{S}\cap \Ih{S'}$.

Moreover, we say $(S,P)\in \SPx{\P}{\xi}$ is \emph{$\lessdot_{\xi}$-minimal
w.r.t.\ $P$ if $(S,P)\lessdot_\xi(S',P)$ for all $(S',P)\in \SPx{\P}{\xi}$
with $S\ne S'$.}  An immediate consequence of \Cref{cor:maxpar=1} will be
useful in the following:
\begin{corollary} \label{cor:first-pi=1}
  Let $\P$ be a set of polygons and let $(S,P)\in \SPx{\P}{\xi}$ be
  $\lessdot_{\xi}$-minimal w.r.t.\ $P$.  Then, $\pari{S}{P}=1$.
\end{corollary}

The key observation of this section is that the parities $\pari{S}{P}$ and
$\pari{S'}{P'}$ of two maximal outstretched segments $(S,P)$ and $(S',P')$
of distinct polygon $P$ and $P'$ that are encountered consecutively along a
sweep-line at $\xi$, i.e., consecutively w.r.t.\ the $\lessdot_{\xi}$
order, determine the relative position of $P$ and $P'$ and their
arrangement in the nesting forest $\NF{\P}$.  To this end, we need the
following definition:

\begin{definition}
  Let $\P$ be a set of polygons and
  $(S,P),\,(S',P')\in\SPx{\P}{\xi}$ with $P\ne P'$.  Then, $(S,P)$ and
  $(S',P')$ are \emph{$\lessdot_{\xi}$-adjacent} if there is no
  $(S'',P'')\in \SPx{\P}{\xi}$ with $P''\in \{P,P'\}$, and 
  \[(S,P)\lessdot_{\xi}(S'',P'')\lessdot_{\xi}(S',P') \textnormal{ or }
  	(S',P')\lessdot_{\xi}(S'',P'')\lessdot_{\xi}(S,P).\]
\end{definition}

The following simple
observation allows us to replace sweep-adjacency by
$\lessdot_{\xi}$-adjacency. This has the advantage that
$\lessdot_{\xi}$-adjacency can be easily evaluated for all $\xi\in
\Ih{S}\cap \Ih{S'}$, and thus in particular for the vertices of the
polygons.

\begin{lemma} \label{lem:lessdot-adj}
  Let $\P$ be a set of polygons, and $(S,P),\,(S',P')\in\SPx{\P}{\xi}$ be
  $\lessdot_{\xi}$-adjacent with $P\ne P'$.  Then, $(S,P)$ and $(S',P')$
  are sweep-adjacent.  In particular, for a sufficiently small $\delta>0$,
  the $(S,P)$ and $(S',P')$ are also $\lessdot_{\xi'}$-adjacent with
  $\xi'\coloneqq \xi+\delta$.
\end{lemma}
\begin{proof}
  Let $\P$ be a set of polygons and $(S,P),\,(S',P')\in\SPx{\P}{\xi}$ be
  $\lessdot_{\xi}$-adjacent with $P\ne P'$.  We assume
  w.l.o.g.\ $(S',P')\lessdot_{\xi}(S,P)$.  Now, consider a sufficiently
  small $\delta>0$.  Then, it is easy to verify that $\xi'\coloneqq
  \xi+\delta\in \Io{S}\cap\Io{S'}$ and $\SPx{\P}{\xi}=\SPx{\P}{\xi'}$.
  Moreover, by definition of $\lessdot_{\xi}$-adjacency, there is no
  $(S'',P'')\in \SPx{\P}{\xi}$ with $P''\in \{P,P'\}$ and
  $(S,P)\lessdot_{\xi}(S'',P'')\lessdot_{\xi}(S',P')$.  The latter two
  observations, together with \cref{lem:bel-iff-lessdot}, imply that
  $(S,P)$ and $(S',P')$ are also $\lessdot_{\xi'}$-adjacent.  In
  particular, there is no $(S'',P'')\in \SPx{\P}{\xi'}$ with $P''\in
  \{P,P'\}$ and $\yVal{S}{\xi'}<\yVal{S''}{\xi'}<\yVal{S'}{\xi'}$.  This,
  together with $\xi'\in \Io{S}\cap \Io{S'}$, implies $(S,P)$ and $(S',P')$
  are sweep-adjacent witnessed by $\xi'$.
\end{proof}
Now, we are in the position to show that the $\lessdot_{\xi}$ order
and the parity of two $\lessdot_{\xi}$-adjacent maximally outstretched
segments determines the inclusion order their respective polygons.

\begin{proposition} \label{prop:parity-inclusion}
  Let $\P$ be a set of polygons and $(S',P'),
  (S,P)\in\SPx{\P}{\xi}$ be $\lessdot_{\xi}$-adjacent with $P\ne P'$.  If
  $(S',P')\lessdot_{\xi}(S,P)$, then the following three statements hold:
  \begin{enumerate}[nolistsep]
  \item $\Int(P)\cap\Int(P')=\emptyset$ if and only if
    $\varpi(S',P')=0$ and $\varpi(S,P)=1$.
  \item $\Int(P)\subsetneq\Int(P')$ if and only if
    $\varpi(S,P)=\varpi(S',P')=1$.
  \item $\Int(P')\subsetneq\Int(P)$ if and only if
    $\varpi(S,P)=\varpi(S',P')=0$.
  \end{enumerate}
  In particular, the case $\varpi(S',P')=1$ and $\varpi(S,P)=0$ 
  can never occur.
\end{proposition}
\begin{proof}
  Let $\P$ be a set of polygons and $(S',P'),
  (S,P)\in\SPx{\P}{\xi}$ be $\lessdot_{\xi}$-adjacent with $P\ne P'$, and
  let $(S',P')\lessdot_{\xi}(S,P)$. First, assume that there is a $\xi'\in
  \Ih{S}\cap \Ih{S'}$ such that $\yVal{S}{\xi'}\ne \yVal{S'}{\xi'}$.  Then,
  \Cref{lem:bel-iff-lessdot}, together with $(S',P')\lessdot_{\xi}(S,P)$,
  implies $(S',P')\lessdot_{\xi'}(S,P)$.  Hence,
  $\yVal{S}{\xi'}<\yVal{S'}{\xi'}$, and thus, $(S,P)$ is below $(S',P')$.
  By \Cref{lem:lessdot-adj}, $(S,P)$ and $(S',P')$ are sweep-adjacent, and
  \cref{cor:below-parity2} implies that \textit{(i)} $\pari{S'}{P'}=0$ and
  $\pari{S}{P}=1$, or \textit{(ii)} $\pari{S'}{P'}=\pari{S}{P}=1$ and
  $\area{P}<\area{P'}$, or \textit{(iii)} $\pari{S'}{P'}=\pari{S}{P}=0$ and
  $\area{P}>\area{P'}$.  Next, assume that we have $\yVal{S}{\xi'}=
  \yVal{S'}{\xi'}$ for all $\xi'\in \Ih{S}\cap \Ih{S'}$.  Then, in
  particular, we have $\yVal{S}{\xi}=\yVal{S'}{\xi}$ and
  $\slope{S}{\xi}=\slope{S'}{\xi}$.  Hence, by definition of
  $\lessdot_{\xi}$, the same three cases \textit{(i)}, \textit{(ii)}, and
  \textit{(iii)} must hold according to \cref{def:ONFP}(3a,3b,3c). In
  summary, if $(S',P'), (S,P)\in\SPx{\P}{\xi}$ are
  $\lessdot_{\xi}$-adjacent and $P\ne P'$, then one of the three
  alternatives \textit{(i)}, \textit{(ii)}, and \textit{(iii)} holds.
  
  Case \textit{(i)}, together with \Cref{lem:adj-overlap}, implies
  $\Int(P)\cap \Int(P')=\emptyset$.  Moreover, Case~\textit{(ii)} and
  Case~\textit{(iii)}, together with \cref{cor:area}, implies
  $\Int(P)\subsetneq \Int(P')$ and $\Int(P')\subsetneq \Int(P)$,
  respectively.  
  Conversely, we can use contraposition and the fact that one of the cases
  \textit{(i)}, \textit{(ii)}, and \textit{(iii)} needs to be satisfied.
  Assume that Case~\textit{(i)} is not satisfied.  Hence,
  Case~\textit{(ii)} or Case~\textit{(iii)} must hold and, as shown above,
  $\Int(P)\subsetneq \Int(P')$ or $\Int(P')\subsetneq \Int(P)$ and thus,
  $\Int(P')\cap\Int(P)\neq \emptyset$.  Thus, Statement (1) is satisfied.
  By similar arguments one shows that Statements (2) and (3) are satisfied.
\end{proof}

\begin{corollary}\label{cor:parity-incl-1}
  Let $\P$ be a set of polygons and assume that $(S',P'),
  (S,P)\in\SPx{\P}{\xi}$ are $\lessdot_{\xi}$-adjacent where $P\ne P'$.  If
  $(S',P')\lessdot_{\xi}(S,P)$, then the following three statements
  are satisfied.
  \begin{enumerate}[nolistsep,noitemsep]
  \item $\Int(P)\subsetneq\Int(P')$ if and only if $\varpi(S',P')=1$.   
  \item $\Int(P')\subsetneq\Int(P)$ if and only if $\varpi(S,P)=0$.
  \item $\Int(P')\nsubseteq\Int(P)$ if and only if $\varpi(S,P)=1$.
  \end{enumerate} 
\end{corollary}
\begin{proof}
  Note that $\varpi(S',P')=1$ precisely in Case (2) of
  \cref{prop:parity-inclusion}, and $\varpi(S,P)=0$ precisely in Case (3)
  of \cref{prop:parity-inclusion} which establishes Statements (1) and (2)
  in this corollary, respectively.  Moreover, negation of Statement (2),
  together with $P\neq P'$, implies Statement (3).
\end{proof}

\section{From Ordered Maximal Outstretched Segments to the Nesting Forest} 
\label{sec:nest-forest}

We start with some simple properties of the nesting forest $\NF{\P}$.  For
fixed $\xi$, set $\P_{|\xi} \coloneqq\{P\in\P\mid \xi\in \Ih{P}\}$ and
denote by $\NF{\P_{|\xi}}$ the nesting forest of this subset of polygons.

\begin{lemma}\label{lem:find-root}
  Let $\P$ be a set of polygons, $P,P'\in \P$, and $(S,P)\in
  \SPx{\P}{\xi}$.  If $\Int(P)\subsetneq \Int(P')$, then there is a
  $(S',P')\in \SPx{\P}{\xi}$ with $(S',P')\lessdot_{\xi}(S,P)$.
\end{lemma}
\begin{proof}
  Let $\P$ be a set of polygons, and let $P,P'\in \P$.  Moreover, let
  $(S,P)\in \SP{\P}$ and $\xi\in \Ih{S}$.  Suppose that $\Int(P)\subsetneq
  \Int(P')$ and, therefore, $P\neq P'$.  Assume, for contradiction, that
  $(S,P)\lessdot_{\xi}(S',P')$ for all $(S',P')\in\SPx{\P}{\xi}$.  First,
  if there is no such $(S',P')\in \SPx{\P}{\xi}$, then $\xi\in
  \Ih{P}\setminus \Ih{P'}$.  Hence, $\Ih{P}\nsubseteq \Ih{P'}$ implies by
  \cref{obs:I-P-sub} that $\Int(P)\nsubseteq\Int(P')$; a contradiction.
  Now, assume that there is a $(S',P')\in \SPx{\P}{\xi}$, and suppose that
  $(S',P')$ is $\lessdot_{\xi}$-minimal w.r.t.\ $P'$.  Then, by
  \cref{cor:first-pi=1}, $\pari{S'}{P'}=1$, and by assumption
  $(S,P)\lessdot_{\xi}(S',P')$.  The latter implies that there is a
  $\lessdot_{\xi}$-maximal $(S,P)\in\SPx{\P}{\xi}$ w.r.t.\ $P$ such that
  $(S,P)\lessdot_{\xi}(S',P')$.  Hence, by the choice of $(S,P)$ and
  $(S',P')$, we conclude that they are $\lessdot_{\xi}$-adjacent.  Since
  $(S,P)\lessdot_{\xi}(S',P')$, the roles of $P$ and $P'$ in
  \cref{cor:parity-incl-1} are switched, and we can use
  \cref{cor:parity-incl-1}~(3) and $\pari{S'}{P'}=1$ to infer
  $\Int(P)\nsubseteq\Int(P')$; a contradiction.  Hence, there must a
  $(S',P')\in \SPx{\P}{\xi}$ with $(S',P')\lessdot_{\xi}(S,P)$.
\end{proof}

\begin{lemma} \label{lem:minmaxP}
  Let $\P$ be a set of polygons and let $(S,P)\in\SPx{\P}{\xi}$. If $(S,P)$
  is the $\lessdot_{\xi}$-minimal element, then $P$ is a root vertex in the
  nesting forest $\NF{\P_{|\xi}}$.
\end{lemma}
\begin{proof}
  Let $\P$ be a set of polygons and let $(S,P)\in\SPx{\P}{\xi}$ be the
  $\lessdot_{\xi}$-minimal element.  Using contraposition, we assume that
  $P$ is \emph{not} a root vertex in the nesting forest $\NF{\P_{|\xi}}$;
  i.e., there is a $P'\in \P$ with 
  $\Int(P)\subsetneq \Int(P')$.
  Then, \Cref{lem:find-root} implies that $(S,P)$ cannot be the
  $\lessdot_{\xi}$-minimal element of $\SPx{\P}{\xi}$.
\end{proof}

\begin{remark}\label{rem:unique-pred}
  Let $\P$ be a set of polygon and $(S,P),(S',P')\in \SPx{\P}{\xi}$.  Then,
  $(S',P')$ is \emph{the $\lessdot_{\xi}$-predecessor} of $(S,P)$ if
  $(S',P')\lessdot_{\xi}(S,P)$ and there is no $(S'',P'')\in \SPx{\P}{\xi}$
  with $(S',P')\lessdot_{\xi}(S'',P'')\lessdot_{\xi}(S,P)$.
Since $\lessdot_{\xi}$ is a strict total order on $\SPx{\P}{\xi}$, the
$\lessdot_{\xi}$-predecessor is always well-defined.
\end{remark}
For instance, in \cref{fig:sweep-xample}, we have the following:
$(S_9,P_3)$ is the $\lessdot_{\xi_2}$-predecessor of $(S_8,P_4)$,
$(S_{11},P_3)$ is the $\lessdot_{\xi_3}$-predecessor of $(S_8,P_4)$,
$(S_4,P_3)$ is the $\lessdot_{\xi_1}$-predecessor of $(S_3,P_2)$
$(S_8,P_5)$ is the $\lessdot_{\xi_2}$-predecessor of $(S_7,P_5)$.
Furthermore, there is no $\lessdot_{\xi_i}$-predecessor of $(S_{12},P_3)$
for $i \in \{1,\ldots,4\}$.  The first two cases demonstrate that the
$\lessdot_{\xi}$-predecessor of a given segments is not independent by the
choice of $\xi$. Below, we will make use of the following observation
regarding $\lessdot_{\xi}$-predecessor:
\begin{lemma} \label{lem:psi} 
  Let $\P$ be a set of polygons and $(S,P),(S',P')\in \SPx{\P}{\xi}$, and
  let $(S',P')$ be the $\lessdot_{\xi}$-predecessor of $(S,P)$, and $P\ne
  P'$.  Then, for every $P''\in \P\setminus\{P,P'\}$, we have
  $\Int(P)\subsetneq \Int(P')$ if and only if $\Int(P')\subsetneq
  \Int(P'')$.
\end{lemma}
\begin{proof}
  Let $\P$ be a set of polygons and $(S,P),(S',P')\in \SPx{\P}{\xi}$, and
  let $(S',P')$ be the $\lessdot_{\xi}$-predecessor of $(S,P)$, and $P\ne P'$.
  Moreover, let $P''\in \P\setminus\{P,P'\}$.
  
  First, assume that $\Int(P)\subsetneq \Int(P'')$. Hence,
  \cref{lem:find-root} implies that there is a $(S'',P'')\in \SPx{\P}{\xi}$
  with $(S'',P'')\lessdot_{\xi}(S,P)$.  In particular, we can assume
  w.l.o.g.\ that $(S'',P'')$ is $\lessdot_{\xi}$-maximal such that
  $(S'',P'')\lessdot_{\xi}(S,P)$.  By \cref{rem:unique-pred} and since
  $(S',P')$ is the $\lessdot_{\xi}$-predecessor of $(S,P)$, we have
  $(S'',P'')\lessdot_{\xi} (S',P') \lessdot_{\xi}(S,P)$.  The latter
  arguments also imply that there is no pair $(\ast, P'')$ with
  $(S'',P'')\lessdot_\xi(\ast,P'')\lessdot_{\xi} (S',P')$.  Let $M$ be
  the set of all pairs $(\ast,P')$ with $(\ast,P')=(S',P')$ or
  $(S'',P'')\lessdot_{\xi} (\ast,P') \lessdot_{\xi} (S',P')$.  Note $M\neq
  \emptyset$ and, by \cref{rem:unique-pred}, there is a
  $\lessdot_{\xi}$-minimal element $(\widetilde{S'},P')$ in $M$.  Taken the
  latter arguments together, $(S'',P'')$ and $(\widetilde{S'},P')$ are
  $\lessdot_{\xi}$-adjacent.  Recall that $(S'',P'')\lessdot_{\xi} (S',P')
  \lessdot_{\xi}(S,P)$.  Hence, by similar arguments, we can choose
  $(\widetilde{S},P)$ as the $\lessdot_{\xi}$-minimal pair among all pairs
  $(\ast, P)\in\SPx{\P}{\xi}$ for which $(\ast, P)=(S,P)$ or
  $(S'',P'')\lessdot_{\xi}(\ast, P)\lessdot_{\xi} (S,P)$.  Again,
  $(S'',P'')$ and $(\widetilde{S},P)$ are $\lessdot_{\xi}$-adjacent with
  $(S'',P'')\lessdot_{\xi}(\widetilde{S},P)$.  This, together with
  \cref{cor:parity-incl-1}~(1) and $\Int(P)\subsetneq \Int(P'')$, implies
  $\varpi(S'',P'')=1$.  Since $(S'',P'')$ and $(\widetilde{S'},P')$ are
  $\lessdot_{\xi}$-adjacent with
  $(S'',P'')\lessdot_{\xi}(\widetilde{S'},P')$,
  \cref{cor:parity-incl-1}~(1) and $\varpi(S'',P'')=1$ imply that
  $\Int(P')\subsetneq \Int(P'')$.

  Conversely, assume that $\Int(P')\subsetneq \Int(P'')$.  By similar
  arguments as used in the previous case, there is a
  $\lessdot_{\xi}$-maximal (w.r.t.\ $P''$) element
  $(S'',P'')\in \SPx{\P}{\xi}$ such that
  $(S'',P'')\lessdot_{\xi}(S',P')\lessdot_{\xi}(S,P)$.  Now, we choose
  $(\widetilde{S},P)$ (resp., $(\widetilde{S'},P')$) as the
  $\lessdot_{\xi}$-minimal pair among all pairs $(\ast, P)\in\SPx{\P}{\xi}$
  (resp., $(\ast, P')\in\SPx{\P}{\xi}$) that satisfy
  $(S'',P'')\lessdot_{\xi}(\widetilde{S},P)$ (resp.,
  $(S'',P'')\lessdot_{\xi}(\widetilde{S'},P')$).  Analogously, and by
  interchanging the roles of $P$ and $P'$, one shows that
  $\Int(P)\subsetneq \Int(P'')$.
\end{proof}

Now, we translate \cref{prop:parity-inclusion} to the relative location of
the polygons in the nesting forest $\NF{\P_{|\xi}}$.
\begin{lemma} \label{lem:Tchildsib}
  Let $\P$ be a set of polygons and $(S,P),(S',P')\in \SPx{\P}{\xi}$. 
  If $(S',P')$ is the $\lessdot_{\xi}$-predecessor of
  $(S,P)$ and $P\ne P'$, then we have the following:
  \begin{enumerate}[nolistsep,noitemsep]
  \item $P$ and $P'$ are siblings in $\NF{\P_{|\xi}}$ if and only if
    $\varpi(S',P')=0$ and $\varpi(S,P)=1$.
  \item $P$ is a child of $P'$ in $\NF{\P_{|\xi}}$ if and only if 
    $\varpi(S,P)=\varpi(S',P')=1$
  \item $P'$ is a child of $P$ in $\NF{\P_{|\xi}}$ if and only if
    $\varpi(S,P)=\varpi(S',P')=0$.
  \end{enumerate}
\end{lemma}
\begin{proof}
  Let $\P$ be a set of polygons and $(S,P),(S',P')\in \SPx{\P}{\xi}$.
  Moreover, suppose the $\lessdot_{\xi}$-predecessor of $(S,P)$ is $(S',P')$ 
  and $P\ne P'$.  In particular, $(S',P')$ and $(S,P)$ are
  $\lessdot_{\xi}$-adjacent.  By \cref{prop:parity-inclusion}, precisely
  one the cases holds: \textit{(i)}~$\pari{S'}{P'}=0$ and $\pari{S}{P}=1$,
  or \textit{(ii)}~$\pari{S}{P}=\pari{S'}{P'}=1$, or
  \textit{(iii)}~$\pari{S}{P}=\pari{S'}{P'}=0$.  Hence, we consider first
  the \emph{if} directions for all Statements~(1) to (3).
  
  In Case~\textit{(i)}, \cref{prop:parity-inclusion}~(1) implies that
  $\Int(P)\cap \Int(P')=\emptyset$.  
  Thus, neither $P\preceq_{ \NF{\P_{|\xi}} }P'$ nor
  $P'\preceq_{ \NF{\P_{|\xi}} }P$ can hold.
  Now, assume for contradiction that $P$
  and $P'$ are \emph{not} siblings in $\NF{\P_{|\xi}}$; i.e., $P$ and $P'$
  are neither both roots nor have a common parent in $\NF{\P_{|\xi}}$.
  Hence, we may assume w.l.o.g.\ that $P$ is not a root in
  $\NF{\P_{|\xi}}$.
  Thus, there is a $P''\in \P\setminus\{P,P'\}$ that is
  the parent of $P$ in $\NF{\P_{|\xi}}$.  In particular, $\Int(P)\subsetneq
  \Int(P'')$.  This, together with \Cref{lem:psi}, implies
  $\Int(P')\subsetneq \Int(P'')$.  Since by assumption, $P''$ cannot be the
  parent of $P'$ in $\NF{\P_{|\xi}}$, we have a $P'''\in
  \P\setminus\{P,P'\}$ such that $\Int(P')\subsetneq \Int(P''')\subsetneq
  \Int(P'')$.  Hence, \Cref{lem:psi} implies $\Int(P)\subsetneq
  \Int(P''')\subsetneq \Int(P'')$; a contradiction to the assumption that
  $P''$ is the parent of $P$ in $\NF{\P_{|\xi}}$.  Therefore, $P$ and $P'$
  are siblings in $\NF{\P_{|\xi}}$.
 
  In Case~\textit{(ii)}, \cref{prop:parity-inclusion}~(2) implies that
  $\Int(P)\subsetneq \Int(P')$.  Now, assume for contradiction that $P$ is
  \emph{not} a child of $P'$ in $\NF{\P_{|\xi}}$.  Hence, there must be a
  $P''\in \P\setminus\{P,P'\}$ such that $\Int(P)\subsetneq
  \Int(P'')\subsetneq\Int(P')$.  Hence, $\Int(P)\subsetneq \Int(P'')$ and
  $\Int(P')\nsubseteq \Int(P'')$, together with \cref{lem:psi}, yields a
  contradiction.  Hence, $P$ is a child of $P'$ in $\NF{\P_{|\xi}}$.

  In Case~\textit{(iii)}, \cref{prop:parity-inclusion}~(3) implies that
  $\Int(P')\subsetneq \Int(P)$.  One can show analogously to
  Case~\textit{(ii)} that $P'$ is a child of $P$ in $\NF{\P_{|\xi}}$.

  Conversely, we show the \emph{only-if} directions for Statements~(1)
  to~(3) by contraposition, and apply the fact that one of the cases
  \textit{(i)}, \textit{(ii)} and \textit{(iii)} needs to be satisfied.
  Assume that Case~\textit{(i)} is not satisfied.  Hence,
  Case~\textit{(ii)} or Case~\textit{(iii)} must hold and, as shown above,
  $P$ is a child of $P'$ in $\NF{\P_{|\xi}}$ or $P'$ is a child of $P$ in
  $\NF{\P_{|\xi}}$; and thus, $P$ and $P'$ cannot be siblings in
  $\NF{\P_{|\xi}}$.  Thus, Statement~(1) is satisfied. By similar arguments
  one shows that Statement~(2) and~(3) are satisfied.
\end{proof}

\begin{corollary} \label{cor:doit}
  Let $\P$ be a set of polygons and let $(S,P)\in\SPx{\P}{\xi}$ be
  $\lessdot_{\xi}$-minimal w.r.t.\ the polygon $P$. Then, exactly one the
  following three statements is true:
  \begin{enumerate}[nolistsep,noitemsep]
    \item $(S,P)$ is $\lessdot_{\xi}$-minimal in $\SPx{\P}{\xi}$ and
      $P$ is a root in $\NF{\P_{|\xi}}$;
    \item The $\lessdot_{\xi}$-predecessor $(S',P')$ of $(S,P)$ satisfies
      $\pari{S'}{P'}=1$ and $P$ is a child of $P'$;
    \item The $\lessdot_{\xi}$-predecessor $(S',P')$ of $(S,P)$ satisfies
      $\pari{S'}{P'}=0$ and $P$ is a sibling of $P'$.
  \end{enumerate}
\end{corollary}
\begin{proof}
    In the first case, $(S,P)$ has no $\lessdot_{\xi}$-predecessor in
    $\SPx{\P}{\xi}$, and \Cref{lem:minmaxP} implies that $P$ is a root in
    $\NF{P_{\xi}}$. Otherwise, $(S,P)$ has a unique
    $\lessdot_{\xi}$-predecessor $(S',P')$, which satisfies $P'\ne P$ since
    $(S,P)$ is $\lessdot_{\xi}$-minimal w.r.t.\ $P$ by assumption. In
    particular, \Cref{cor:first-pi=1} implies $\pari{S}{P}=1$. Therefore,
    we conclude by \Cref{lem:Tchildsib} that $P$ is as child of $P'$ if
    $\pari{S'}{P'}=1$ and $P$ is a sibling if $\pari{S'}{P'}=0$. Clearly,
    the three alternatives are mutually exclusive and cover all possible
    cases.
\end{proof}

\Cref{cor:doit} can be used to construct the nesting forest
$\NF{\P_{|\xi}}$ by means of a single traversal of $\SPx{\P}{\xi}$ in
$\lessdot_{\xi}$-order. A forest is uniquely determined by the map
$\parent\colon V\to V\cup\{\emptyset\}$ that marks each root $r\in V$ by
$\parent(r)=\emptyset$ and assigns to every other vertex $v$ its parent. 

\begin{lemma} \label{lem:doit-parent}
  Let $\P$ be a set of polygons.  The following procedure correctly
  determines the parent-function for all vertices $P$ in $\NF{\P_{|\xi}}$:
  \begin{itemize}[noitemsep,leftmargin=*]
  \item[] \textbf{For all} $(S,P)\in\SPx{\P}{\xi}$ from small to
    large w.r.t.\ $\lessdot_{\xi}$ \textbf{do}
    \begin{enumerate}[noitemsep,leftmargin=*]
    \item[] \textbf{If} $(S,P)$ is 	$\lessdot_{\xi}$-minimal  w.r.t.\ $P$ 
      \textbf{then}
      \begin{enumerate}[noitemsep]
      \item \textbf{If} $(S,P)$ is $\lessdot_{\xi}$-minimal in 
	$\SPx{\P}{\xi}$ \textbf{then} $\parent(P)\gets\emptyset$ 
      \item \textbf{If} the $\lessdot_{\xi}$-predecessor $(S',P')$ of $(S,P)$
	in $\SPx{\P}{\xi}$ satisfies
        $\pari{S'}{P'}=1$, \textbf{then} $\parent(P)\gets P'$.
      \item \textbf{If} the $\lessdot_{\xi}$-predecessor $(S',P')$ of $(S,P)$ 
        in $\SPx{\P}{\xi}$ satisfies $\pari{S'}{P'}=0$, 
        \textbf{then} $\parent(P)\gets \parent(P')$. 
      \end{enumerate}	
    \end{enumerate}
  \end{itemize}
\end{lemma}
\begin{proof}
  The unique $\lessdot_{\xi}$-minimal segment $(S_0,P_0)\in\SPx{\P}{\xi}$
  identifies a root $\rho_{0}$ of $\NF{\P_{|\xi}}$ by \cref{cor:doit}~(1),
  and thus $\parent(\rho_{0})$ is correctly set to $\emptyset$.
  Then, every other element of $\SPx{\P}{\xi}$ has a unique predecessor;
  say $(S',P')$.  If $\pari{S'}{P'}=1$, then by \cref{cor:doit}~(2) we
  have $\parent(P)=P'$, i.e., it is correctly assigned in Case (b).

  Otherwise, $\pari{S'}{P'}=0$ implies by \cref{cor:doit}~(3) that $P$ is a
  sibling of $P'$, and thus $\parent(P)=\parent(P')$.  Since we
  traverse $\SPx{\P}{\xi}$ in $\lessdot_{\xi}$-order, the
  $\lessdot_{\xi}$-minimal pair $(\widetilde{S'},P')$ w.r.t.\ to $P'$ has
  been processed before $(S,P)$. Since every sequence of sibling
  relationships starts either with the root $P_0$ or with a polygon $P_1$
  for which $\parent(P_1)$ is determined directly according to Case~(b) it
  follows that $\parent(P')$ is already known explicitly when $(S,P)$ is
  processed. Therefore, $\parent(P)\gets\parent(P')$ correctly assigns
  the information on the parent of $P$.  A single traversal of
  $\SPx{\P}{\xi}$ in $\lessdot_{\xi}$-order is therefore sufficient
  to determine $\parent(P)$ for all $P\in\P_{|\xi}$.
\end{proof}

It remains to show that $\NF{\P_{|\xi}}$ can be extended to the complete
forest $\NF{\P}$. The following result shows that this can be achieved by
a set of ``sweep-lines'' $\xi$ such that every polygon is met at least
once by a sweep-line.
\begin{proposition} \label{thm:main1}
  The nesting forest $\NF{\P}$ is computed correctly by applying the
  procedure outlined in \Cref{lem:doit-parent} for a finite set of values
  $X\subset\mathbb{R}$ such that for every $P\in\P$ there is $\xi\in
  \Ih{P}\cap X$. In particular, it suffices to choose 
  $X\coloneqq \{\min(\Ih{P}) \mid P\in\P\}$. 
\end{proposition}
\begin{proof}
  By \cref{obs:I-P-sub}, $\Int(P)\subseteq \Int(P')$ implies
  $\Ih{P}\subseteq \Ih{P'}$.  Thus, if $\xi\in \Ih{P}$, then $\xi\in
  \Ih{P'}$ for all ancestors of $P$ in $\parent_{\NF{\P}}(P)$, and thus
  $\parent_{\NF{\P}}(P)=\parent_{\NF{\P_{|\xi}}}(P)$ for all $\xi\in
  \Ih{P}$. Since, by assumption, a sweep-line $\xi$ is employed such that
  $\xi\in \Ih{P}$ for every $P\in\P$, $\parent(P)$ is correctly determined
  for all $P\in\P$. In particular, therefore, it suffices for each polygon
  $P$ to consider only the ``first appearance'' of $P$, i.e.,
  $\min(\Ih{P})$.
\end{proof}

\Cref{thm:main1} suggests to ``sweep'' along the $x$-axis. The consistence
of the $\lessdot_{\xi}$-order for different values of $\xi$,
\Cref{lem:bel-iff-lessdot}, furthermore, indicates that it is not necessary
to determine the $\lessdot_{\xi}$-order again for each $\xi$.  Instead, the
order $\lessdot_{\xi}$ can be reused for the next position $\xi'$ of the
sweep-line. However, since $\SPx{\P}{\xi}$ and $\SPx{\P}{\xi'}$ differ, it
will be necessary to update $\SPx{\P}{\xi}$ by adding the the maximally
outstretched segments $(S,P)$ with $\xi<\min(S)_x\le\xi'$ and by removing
those with $\max(S)_x\le\xi'$.  Note that the removal is necessary since
otherwise the $\lessdot_{\xi'}$-predecessor cannot be computed
correctly. Taken together, these observations imply that it suffices to
consider as sweep-line positions exactly the $x$-coordinates of the
terminal vertices of the maximal outstretched segments, i.e., the set
$\{\xi \mid \xi = \min(S)_x \text{ or } \xi=\max(S)_x \textnormal{ with
}(S,P)\in\SP{\P}\}$ and to update $\SPx{\P}{\xi}$ and $\lessdot_{\xi}$
exactly at these positions by inserting $(S,P)$ at $\xi=\min(S)_x$ and
removing $(S,P)$ at $\xi=\max(S)_x$. Recalling, for fixed $\xi$, insertion
must follow the $\lessdot_{\xi}$ order to ensure that the $\lessdot_{\xi}$
predecessor function is evaluated correctly. This defines an order $\ble$
in which maximal outstretched segments have to be inserted:
\begin{definition}
  Let $\P$ be a set of polygons. Then, we define the \emph{insertion order}
  $\ble$ on $\SP{\P}$ by setting $(S,P)\ble(S',P')$ whenever
  \begin{enumerate}[nolistsep,noitemsep]
  \item $\min(S)_x<\min(S')_x$, or
  \item $\min(S)_x=\min(S')_x$ and $(S,P)\lessdot_{\min(S)_x}(S',P')$.
  \end{enumerate}
\end{definition}
Since $\min(S)_x=\min(S')_x$ implies $(S,P),(S',P')\in\SPx{\P}{\min(S)_x}$
and $\lessdot_{\xi}$ is a strict total order in $\SPx{\P}{\xi}$, it is
clear that~$\ble$ is indeed a strict total order on $\SP{\P}$.
\begin{algorithm}[tbp]
  \caption{\texttt{NestingForest($\P$)}}
  \label{alg:nesting}
  \begin{algorithmic}[1]
    \Require Set $\P$ of overlap-free polygons
    \Ensure  Nesting forest $\NF{P}$
    \For {\textbf{each} $P\in\P$}
       \State compute maximal outstretched segments $S$ of $P$ and 
              determine $\varpi(S)$ \label{lin:mos}
       \State compute $\area{P}$ \label{lin:area}
    \EndFor
    \State $\mathcal{L} \leftarrow$ \textbf{sort} 
           $\SP{\P}$ w.r.t.\ insertion order $\ble$
           \Comment{insertion list} \label{lin:io}
    \State $\first(S,P)\gets true$ whenever
           $(S,P)\in \mc{L}$ is $\ble$-minimal w.r.t.\ $P$ \label{lin:fp}
    \State $\mathcal{R} \leftarrow$ \textbf{sort} 
           $\SP{\P}$ w.r.t.\ increasing values of $\max(S)_x$
           \Comment{removal list} \label{lin:do}
    \State $\mathcal{R}\cup \mathcal{L} \leftarrow$ \textbf{merge}
           $\mc{L}$ and $\mc{R}$
           w.r.t.\ $\max(S)_x$ for $(S,P)\in\mathcal{R}$ and 
           $\min(S')_x$ for $(S',P')\in\mathcal{L}$ \label{lin:merge}
    \State \BST$\leftarrow\emptyset$ \label{lin:BST-init}
           \Comment{initialize balanced search tree for $\lessdot_{\xi}$}
    \For{ \textbf{each} $(S,P)$ in the order of $\mathcal{R}\cup\mathcal{L}$}
      \label{lin:for}
      \If {$(S,P)$ taken from $\mathcal{R}$} 
      	delete $(S,P)$ from \BST \label{lin:BST-del}
      \Else{} 
        insert $(S,P)$ into \BST \textbf{and} \label{lin:BST-ins}
        \Comment{$(S,P)$ taken from $\mathcal{L}$}
        \If {$\first(S,P)$} \label{lin:if-first}
          \If {$(S,P)$ is $\lessdot_{\xi}$-minimum in \BST}
            $\parent(P)\leftarrow\emptyset$
            \Comment{insert $P$ into $F$ as root of a subforest}
            \label{lin:if-sec}
          \Else{} 
            find $\lessdot_{\xi}$-predecessor $(S',P')$
            of $(S,P)$ in \BST \textbf{and} \label{lin:BST-search}
            \Comment{$(S,P)$ has a $\lessdot_{\xi}$ predecessor}
            \If {$\pari{S'}{P'}=1$}
              $\parent(P)\leftarrow P'$ \label{lin:if-third}
            \Else  
              { }$\parent(P)\leftarrow \parent(P')$
              \label{lin:endfor}  
              \Comment{$P$ and $P'$ are siblings}
              \label{lin:if-third2}
            \EndIf
          \EndIf               
        \EndIf 
      \EndIf 
    \EndFor  
    \State \Return $\mathscr{F}$ \Comment{as defined by $\parent(\,\cdot\,)$}
  \end{algorithmic}
\end{algorithm}
Now, combining the considerations at the beginning of this section and the
results above, we arrive at the algorithm summarized in \Cref{alg:nesting}.
\begin{theorem}\label{lem:alg-correct}
  \Cref{alg:nesting} correctly computes the nesting forest of a set $\P$ of
  overlap-free polygons. 
\end{theorem}
\begin{proof}
  By definition of the insertion order $\ble$ and the ordered insertion
  lists $\mathcal{L}$ and $\mathcal{R}$, all $(S,P)$ with the same values
  $\min(S)_x=\xi$ are inserted consecutively into $\BST$ and all $(S,P)$
  with the same values $\max(S)_x=\xi$ are removed consecutively from
  $\BST$.  Moreover, for given $\xi$, removal happens before insertions.
  Therefore, $\BST$ correctly holds the order $\lessdot_{\xi}$ of
  $\SPx{\P}{\xi}$ before and after processing $\xi$. The insertion order
  $\ble$ prescribes that, for given $\xi$, the $\lessdot_{\xi}$-smaller
  elements $(S,P)\in\SPx{\P}{\xi}$ are inserted first. Thus, in particular,
  every $(S,P)$ is inserted after its $\lessdot_{\xi}$-predecessor
  $(S',P')$, and thus the $\lessdot_{\xi}$-predecessor is correctly
  determined in \BST. Moreover, for every polygon $P$, the $(S,P)$ the
  first segment to be processed is the one that has minimal value of
  $\min(S)_x$, and is $\lessdot_{\min(S)_x}$-minimal.  Since it is either
  $\lessdot_{\xi}$-minimal (and thus the polygon $P$ is the root of
  subforest of $F$) or the $\lessdot_{\xi}$-predecessor $(S',P')$ exists
  and belongs to a different polygon $P'\ne P$. In the latter case,
  \Cref{lem:doit-parent} implies that $P$ is either a child or a sibling
  $P'$, depending whether $\pari{S'}{P'}=1$ and $\pari{S'}{P'}=0$, which in
  turn implies $\parent(P)=P'$ or $\parent(P)=\parent(P')$, respectively.
  Thus, \Cref{alg:nesting} correctly identifies the parent of every
  $P\in\P$. Since $\NF{\P}$ is uniquely determined by its vertex set and
  parent function, \Cref{alg:nesting}, correctly computes the nesting
  forest $\NF{\P}$.
\end{proof}

Note that the insertion order $\ble$ and the lists $\mathcal{L}$ and
$\mathcal{R}$ of insertions and removal of $(S,P)$ do not need to be
represented explicitly. In practice, it suffices to determine the set
$X\coloneqq \bigcup_{(S,P)\in\SP{\P}}\{\min(S)_x,\max(S)_x\}$ of sweep-line
positions. For each $\xi\in X$, one first removes all $(S,P)$ with
$\max(S)_x=\xi$ from \BST, then sorts the $(S,P)$ with $\min(S)_x=\xi$
w.r.t.\ $\lessdot_{\xi}$ and proceeds to insert them in this order. Marking
a polygon $P$ as ``seen'' when it appears for the first time in
$\ble$-order can be used instead of precomputing $\first(S,P)$. Therefore,
it is possible to implement \Cref{alg:nesting} in a single pass of the
sweep-line position $\xi$ instead of precomputing the order
$\ble$. However, this does not affect the asymptotic running time.

Now, let us turn to analyzing the running time of \Cref{alg:nesting}.  In
the following, we write $m\coloneqq|\P|$ for the number of polygons,
$n\coloneqq \sum_{P\in\P}|V(P)|$ for the total number of vertices, and
$N\coloneqq |\SP{\P}|$ for the total number of maximal outstretched
segments.

\begin{lemma} \label{lem:area-effort}
  The total effort for computing $\area{P}$ for all $P\in\P$ is
  $\mc{O}(n)$.
\end{lemma}
\begin{proof}
  The area of a simple polygon can be computed efficiently using the
  ``shoelace formula'', also known as ``Gau{\ss}'s area formula''
  \cite[p.\ 53]{GaussBd12} and ``Surveyor's formula'' \cite{Meister1769}.
  As shown in \cite{Braden1986}, the effort for a single polygon is
  $\mc{O}(|V(P)|)$, and thus the total effort is $\mc{O}(n)$. 
\end{proof}

\begin{lemma} \label{lem:mos-effort}
  For a set of polygons $\P$, the total effort for computing the set of
  maximal outstretched segments $\SP{\P}$ and the parity function
  $\pari{S}{P}$ for all $(S,P)\in\SP{\P}$ is $\mc{O}(n)$.
\end{lemma}
\begin{proof}
  Consider a polygon $P\in\P$. Starting at an arbitrary vertex $u\in V(P)$,
  set $S=\emptyset$ and traverse $P$ in the order of an arbitrarily chosen
  edge $e$ incident with $u$, and proceed as follows: If the $x$-coordinate
  decreases along $u$ keep adding edges with non-increasing $x$-coordinates
  to $S$ until the first edge with increasing $x$-coordinate is
  found. Otherwise, add edges with non-decreasing $x$-coordinates until the
  first edge with decreasing $x$-coordinate is encountered. Every time an
  edge with opposing directions along the $x$-coordinate is found, a new
  segment $S'$ is started. The procedure stops after $|V(P)|$ edges and
  vertices, when the starting point is $u$ encountered again.  The first
  segment $S$ and last segment $S^*$ are possibly incomplete, in which case
  they are part of the same maximal outstretched segment. If $S$ and $S^*$
  contain non-vertical edges with the same directions, we concatenate them
  at $u$.  Clearly, this can be done with $\mc{O}(|V(P)|)$
  effort. Otherwise, $S$ and $S^*$ are correspond to separate maximal
  outstretched segments. In our construction, $S$ may contain leading
  vertical edges, which can be removed in $\mc{O}(|V(P)|)$ time. Finally,
  all trailing vertical edges are removed from all segments, which
  obviously also can be done in $\mc{O}(|V(P)|$ effort. Thus, the maximal
  outstretched segments of $P$ can be computed in $\mc{O}(|V(P)|)$
  time. Hence, the total effort for computing $\SP{\P}$ is $\mc{O}(n)$.

  While traversing $P$, we can also keep the information which segment
  contains a vertex with maximal $y$-coordinate. Note that such a vertex
  $\hat p$ always exists and is contained in maximal outstretched segments
  $S$.  If the segment $(S,P)\in \SP{\P}$ with $\hat{p}\in V(S)$ is unique,
  then $(S,P)\in \SPx{\P}{\hat{p}_x}$ is minimal w.r.t.\ $P$, and thus, we
  conclude by \cref{cor:first-pi=1} that $\blw{S}{P}$.  Otherwise, $\hat p$
  is a terminal vertex of two maximal outstretched segments $S$ and $S'$
  with $\hat{p}\in V(S)\cap V(S')$.  If $\hat p=\min(S)=\min(S')$ with
  incident edges $e\in S$ and $e'\in S$' such that $\Delta(e)>\Delta(e')$,
  then $\varpi(S,P)=1$ and $\varpi(S,P')=0$.  If $\hat p=\max(S)=\max(S')$
  with incident edges $e\in S$ and $e'\in S$' such that
  $\Delta(e)<\Delta(e')$, then $\varpi(S,P)=1$ and $\varpi(S,P')=0$. The
  effort for determining $\hat p$, storing a pointer to its incident
  segment(s), determining whether $\hat p$ is a terminal vertex, and
  computing the slopes at each vertex is a constant-time overhead for each
  vertex during the traversal of $P$. The total effort for determining
  $\hat p$ and the parity of the incident segments is therefore
  $\mc{O}(|V(P)|)$.  Consecutive maximal outstretched segments are then
  given alternating parity (cf.\ \cref{cor:maxpar=1}), which clearly also
  requires only linear effort.  Thus, the total effort for assigning the
  parity is $\mc{O}(n)$.
\end{proof}

We store the maximal outstretched segments $(S,P)$ explicitly as
separate lists of their edges ordered with increasing $x$-coordinates
$\min(e)_x$. These lists will be required to identify the appropriate
edge $e\in E(S)$ for given $\xi$, i.e., the edge satisfying $\min(e)_x\le
\xi<\max(e)_x$ for $\xi\in \Ih{S}$.
Moreover, assuming that $\area{P}$ is precomputed and can be looked up in constant
time and $\slope{S}{\xi}$ can obviously be evaluated in constant time if
the appropriate edge $e \in E(S)$ with $\xi \in \Ih{e}$ is known, we observe
that each of the conditions for $\lessdot_{\xi}$ in \cref{def:ONFP} can be
checked in constant time:
\begin{obs} \label{obs:lessdotconst}
  After $\mc{O}(n)$ preprocessing effort, the effort to evaluate for
  $(S,P),(S',P')\in\SP{\P}$ with $\Ih{S}\cap \Ih{S'}\ne\emptyset$, whether
  $(S,P)\lessdot_{\xi} (S,P')$, $(S',P')\lessdot_{\xi} (S,P)$, or
  $(S,P)=(S',P')$ is $\mc{O}(1)$ for any $\xi\in \Ih{S}\cap \Ih{S'}$
  if the appropriate edges $e\in E(S)$ with $\xi \in \Ih{e}$ and 
  $e'\in E(S')$ with $\xi\in \Ih{e'}$ are known.
\end{obs}

\begin{lemma}
  For a set $\P$ of polygons, the set $\SP{\P}$ can be sorted in
  ``insertion order'' $\ble$ in $\mc{O}(N \log N)$ time.
\end{lemma}
\begin{proof}
  The minima $\min(S)$ and their $x$-coordinates can be determined in
  constant time for each maximal outstretched segment $(S,P)$.  For all
  $(S,P)$ and $(S',P')$ a comparison w.r.t.\ $\ble$ entails a comparison
  the $\min(S)_x$ and $\min(S')_x$, requires only constant time, and, in
  the case of equality, a comparison w.r.t.\ $\lessdot_{\min(S)_x}$, which
  can be performed in constant time by \Cref{obs:lessdotconst}. Since
  $|\SP{\P}|=N$, the set $\SP{\P}$ can be sorted w.r.t.\ to $\ble$ in $\mc{O}(N
  \log N)$ time using a standard sorting algorithm, e.g.\ heap sort.
\end{proof}

A \emph{self-balancing binary search trees (\BST)} can be used to maintain
the order on $\SPx{\P}{\xi}$. Since $|\SPx{\P}{\xi}|\le N$, the \BST
guarantees search, insertion and deletion of entries in $\mc{O}(\log N)$
time, provided comparisons w.r.t.\ to the relevant total order can be
performed in constant time, see e.g.\ \cite[Sec.\ 6.2.3]{Knuth1998}.
\begin{theorem}\label{thm:alg1-anal}
  \Cref{alg:nesting} can be implemented to run in $\mc{O}(n+N \log N)$
  time and $\mc{O}(n)$ space for any set $\P$ of overlap-free polygons
  with $m=|\P|$, $n=\sum_{P\in\P}|V(P)|$ and $N=|\SP{\P}|$.
\end{theorem}
\begin{proof}
  The identification of all maximal outstretched segments $\SP{P}$, the
  precomputation of the parity functions $\pari{S}{P}$ for all $(S,P)\in
  \SP{\P}$, in \cref{lin:mos}, can be performed in $\mc{O}(n)$ time by
  \Cref{lem:mos-effort}. Moreover, computing the $\area{P}$ for all $P\in
  \P$, in \cref{lin:area}, can be performed in $\mc{O}(n)$ time by
  \Cref{lem:area-effort}.  The insertion order $\ble$, in \cref{lin:io},
  requires $\mc{O}(N \log N)$ time.  In particular, setting the pointer
  ``first'' in \cref{lin:fp}, can be done in $\mc{O}(N)\subseteq\mc{O}(n)$.
  The $x$-coordinates of the maxima $\max(S)$, in \cref{lin:do}, can also
  be sorted on $\mc{O}(N \log N)$ time.  Interleaving the $x$-coordinates
  for insertions and deletions, in \cref{lin:merge}, requires $\mc{O}(N)$
  steps, corresponding to a merging step of the insertion and deletion
  orders.
  
  Next, we consider the total running-time of the \BST-operations,
  i.e., insertions, deletions, and predecessor search, in
  \cref{lin:BST-init,lin:BST-del,lin:BST-ins,lin:BST-search}. First, we
  note that each $(S,P)\in \SP{\P}$ is inserted and removed exactly once
  from the \BST. In addition, the $\lessdot_{\xi}$-predecessor needs to be
  found once for each polygon. This amounts to $2N+m\in\mc{O}(N)$
  \BST-operations.  The \BST data structure guarantees that each operation
  can be performed in $\mc{O}(\log N)$ time \emph{provided} the comparison
  function can be evaluated with constant effort.
  The evaluation of $\lessdot_{\xi}$ in constant time is contingent on the
  availability of the ``appropriate edge'' $e$ of $(S,P)$, i.e., the one
  that satisfies $\xi\in \Ih{e}$, cf.\ \cref{obs:lessdotconst}.  To this
  end, we associate a \emph{current edge} $e_{(S,P)}$ with each $(S,P)$
  that is in \BST.  Recall that the maximal outstretched segments are
  traversed in order during preprocessing. Hence, an ordered list of their
  edges with increasing $\min(e)_x$ is obtain as part of the preprocessing
  in $\mc{O}(n)$ total time. We initialize $e_{(S,P)}$ with the edge $e \in
  E(S)$ that contains $\min(S)$.  When accessing a vertex $(S,P)$ in the
  \BST, we check whether $e_{(S,P)}$ intersects $\xi$. This test can be
  performed in constant time. Since the number of \BST-operations is in
  $\mc{O}(N)$ and, for each \BST-operation, the number of accessed vertices
  is in $\mc{O}(\log N)$, the total number of accessed vertices in \BST is
  in $\mc{O}(N\log N)$.  Therefore, $\mc{O}(N\log N)$ tests for the
  intersection of $\xi$ and $e_{(S,P)}$ are performed. If a test fails, we
  advance along the ordered edges of $(S,P)$ until we reach the edge $e'\in
  E(S)$ with $\xi\in \Ih{e'}$, and update $e_{(S,P)}\gets e'$.  Since each
  edge of $(S,P)$ is traversed at most once during the updates of
  $e_{(S,P)}$, the total effort is $\sum_{(S,P)\in\SP{\P}} |E(S)| \le
  n$. Therefore, the effort for maintaining the current edges is in
  $\mc{O}(n+ N\log N)$. Hence, performing the \BST-operations in
  \cref{lin:BST-init,lin:BST-del,lin:BST-ins,lin:BST-search}, together with
  maintaining the current edges, requires a total effort of $\mc{O}(n+N
  \log N)$. Each of the
  \Cref{lin:if-first,lin:if-sec,lin:if-third,lin:if-third2} require only
  constant time for each $(S,P)$, and thus, the total effort for these
  operations is in $\mc{O}(N)$.
  
  Taken together \Cref{alg:nesting} requires $\mc{O}(n)$ operations for
  preprocessing, $\mc{O}(N\log N)$ effort to construct the insertion order
  $\ble$, and $\mc{O}(n+N\log N)$ effort to maintain the current edges
  required to perform the $\mc{O}(N\log N)$ \BST-operations with
  constant-time $\lessdot_{\xi}$ comparisons.  Therefore, the total running
  time is in $\mc{O}(n+N\log N)$.

  Moreover, saving the necessary information of the polygons (i.e., their
  vertices, edges, area and ``parent'') is in $\mc{O}(n)$ space.  Likewise,
  saving the set $\SP{\P}$ of maximal outstretched segments, together with
  their vertices and edges, is in $\mc{O}(n)$ space.  The \BST and the
  sorted lists $\mc{L}$, $\mc{R}$ and $\mc{R}\cup \mc{L}$ of maximal
  outstretched segments are each in $\mc{O}(N)\subseteq \mc{O}(n)$ space.
  Hence, the total space required is in $\mc{O}(n)$.
\end{proof}

\begin{lemma} \label{lem:lowerbound}
  The polygon nesting problem is in $\Omega(n+m\log m)$.
\end{lemma}
\begin{proof}
  Consider the special case that the nesting forest is a path, i.e, there
  are no siblings. Then, the nesting problem reduced to sorting the
  polygons by size. The sorting-problem of $m$ elements w.r.t.\ to an order
  for which comparisons can be evaluated in constant time is in $\Omega(m
  \log m)$, see e.g.\ \cite[Sec.\ 3.3]{Aho1974TheDA} and
  \cite[Sec.\ 2.1.6]{Mehlhorn1984}. This, together with the fact that the
  nesting-problem is in $\Omega(n)$, implies the lower bound $\Omega(n+m
  \log m)$.
\end{proof}

\Cref{thm:alg1-anal} and \Cref{lem:lowerbound} together imply
\begin{corollary}\label{cor:compl-anal}
  For a set $\P$ of overlap-free polygons with $m=|\P|$,
  $n=\sum_{P\in\P}|V(P)|$ and $N=|\SP{\P}|$, the worst-case time complexity
  of the nesting-problem is in $\Omega(n+m \log m)$ and $\mc{O}(n+N \log
  N)$.
\end{corollary}

A polygon $P$ is called \emph{convex} whenever $\Int(P)$ is a convex set.
It is easy to verify that a convex polygon harbors exactly $2$ maximal
outstretched segments. Hence, \Cref{cor:compl-anal} implies that our
approach is asymptotically optimal in this case.
\begin{corollary}
  Let $\P$ be a set of $m$ overlap-free polygons. Suppose every $P\in\P$ is
  convex or $|V(P)|\le K$ for some fixed constant $K$. Then,
  \Cref{alg:nesting} runs in $\Theta(n+m\log m)$. In particular, this is
  optimal.
\end{corollary}
\begin{proof}
  If $P$ is convex, it contains exactly $2$ maximally outstretched
  segments. Similarly, if $|V(P)|\le K$ then there are $\Theta(1)$
  maximally outstretched segments in each polygon, and thus
  $\mc{O}(N)=\mc{O}(m)$.  Hence, the restricted polygon nesting problem can
  be solved in $\mc{O}(n+m\log m)$ time by \Cref{alg:nesting}. The lower
  bound of \Cref{lem:lowerbound} shows that the running time of
  \Cref{alg:nesting} is asymptotically optimal for this special case.
\end{proof}

\section{Summary and Outlook}

We have described here a variant of \emph{sweep line algorithm} that
  determines the nesting of polygons with non-intersecting interior,
  generalizing a similar algorithm by \citet{Bajaj:90} for non-touching
  polygons. The main innovation of our approach is the definition of
  \emph{maximal outstretches segments} and a corresponding ordering of these
  segments along the sweep line that can be computed and maintained
  efficiently and consistently handles overlapping points and edges along
  these segments. This construction makes it possible to achieve a running
  time of $\mathcal{O}(n+N\log N)$, where $N<n$ is the number of maximal
  outstretched segments. The resulting algorithm is optimal e.g.\ for
  convex polygons. The algorithm of \citet{Bajaj:90} uses ``subchains''
  that are parts of convex chains. These are subsets of the maximal
  outstretched segments introduced here. While this does not, in general,
  yield an asymptotic improvement, it reduced the number of segments that
  have to be considered.

In summary, we computed the nesting of touching polygons. 
  However, we did not determine whether, and in the affirmative case, 
  where exactly two polygons $P,P'$ of $\mc{P}$ are touching. 
  Such points can of course be determined by the classical
  sweep line approach, and then can be added to data on the corners of $\mathcal{P}$.
  It is necessary, however, to order such touching points with different
  incident polygons along each of the edges of the polygon $\mathcal{P}$. 
  We shall consider this problem in more detail in a forthcoming 
  contribution.

A problem closely related to the nesting of polygons is to consider
  the nesting of their connected components. This amounts to considering
  the (connected components of) the the graph $G$ obtained as the union of
  the vertices of edges of the polygons, is insufficient to completely
  specify the set of polygons. The 2-basis comprising of the facets of the
  planar embedding of $G$ \cite{MacLane:37}, in particular, results in
  decomposition of $G$ into non-overlapping polygons such that all polygons
  are siblings. It is worth noting that a 2-basis of minimum total length
  can be computed in linear time \cite{Liebchen:07}.  Nested sets of
  polygons are obtained from a two-basis as hierarchy (w.r.t.\ to
  inclusion) restricted to sets of facets such that the sum of the facets
  in each set forms a simple cycle, i.e., a polygon. The connected
  components of $G$ in the given embedding have (not necessarily simple)
  cycles as their inner and outer limits, which in contrast to the polygons
  considered here, may also contain degenerate points, i.e., the interior
  of these polygons is no longer connected.  Nesting of connected
  components then can be understood in terms of these possibly generate
  polygons. We suspect that it is sufficient for connected component
  nesting to consider the 2-connected components of the outline-cycles,
  which would reduce the problem to the touching simple polygons considered
  here.

\section*{Acknowledgments}
We thank David Schaller for stimulating discussions on this topic.
This work was funded by the German Research Foundation (DFG) (Proj.\ No.\ MI 439/14-2).

\bibliographystyle{elsarticle-harv}
\bibliography{sweep-line}

\end{document}